\documentclass{amsart}
\usepackage{amsmath,amsfonts,amscd,amssymb,epsfig,euscript,bm}
\usepackage{amsxtra,xcolor,titletoc} 
\usepackage{mathrsfs}
\usepackage{tikz}
\usepackage{pgf,tikz-cd}
\usepackage{amsmath,amssymb}
\usetikzlibrary{decorations.markings}
\usepackage{caption,slashed}
\numberwithin{equation}{section}
\usepackage{pdfsync}
\newtheorem{theorem}{Theorem}
\newtheorem*{theorem*}{Theorem}
\newtheorem{proposition}{Proposition}
\newtheorem{lemma}{Lemma}

\theoremstyle{definition}
\newtheorem*{definition}{Definition}{}

\theoremstyle{remark} 
\newtheorem{remark}{Remark}

\newcommand{\field}[1]{\ensuremath{\mathbb{#1}}}
\newcommand{\CC}{\field{C}}

\newcommand{\RR}{\field{R}}

\newcommand{\ZZ}{\field{Z}}
\DeclareMathOperator{\Tr}{\mathrm{Tr}}
\DeclareMathOperator{\im}{Im}

\newcommand{\vp}{\varphi}

\newcommand{\beq}{\begin{equation}\begin{aligned}}
\newcommand{\eeq}{\end{aligned}\end{equation}}

\newcommand{\al}{\alpha}

\newcommand{\be}{\beta}
\newcommand{\del}{\delta}

\newcommand{\pa}{\partial}
\newcommand{\la}{\langle}
\newcommand{\ra}{\rangle}

\newcommand{\ov}{\over}

\newcommand{\curly}[1]{\mathscr{#1}}

\newcommand{\cD}{\curly{D}}

\newcommand{\cH}{\curly{H}}

\newcommand{\cL}{\curly{L}}

\newcommand{\LL}{\mathcal{L}}

\newcommand{\ga}{\gamma}

\newcommand{\ex}[1]{\langle #1 \rangle}
\newcommand{\Pf}[1]{\text{Pf}( #1 )}
\makeatletter
\@namedef{subjclassname@2020}{\textup{2020} Mathematics Subject Classification}
\makeatother

\usepackage{hyperref}
\hypersetup{colorlinks,citecolor=blue,plainpages=false,hypertexnames=false}
\begin{document}
\title{Supersymmetry and trace formulas I. Compact Lie groups}
\author{Changha Choi}
\address{Simons Center for Geometry and Physics, Stony Brook University, Stony Brook, NY 11794, USA}
\author{ Leon A. Takhtajan}
\address{Department of Mathematics,
Stony Brook University, Stony Brook, NY 11794 USA; 
\newline
Euler International Mathematical Institute, Pesochnaya Nab. 10, Saint Petersburg 197022 Russia}
\begin{abstract}
In the context of supersymmetric quantum mechanics we formulate new supersymmetric localization principle, with application to trace formulas for a full thermal partition function. Unlike the standard localization principle, this new principle allows to compute the supertrace of non-supersymmetric observables, and is based on the existence of fermionic zero modes. We describe corresponding new invariant supersymmetric deformations of the path integral; they differ from the standard deformations arising from the circle action and require higher derivatives terms. Consequently, we prove that the path integral localizes to periodic orbits and not only on constant ones. We illustrate the principle by deriving bosonic trace formulas on compact Lie groups, including classical Jacobi inversion formula.

\end{abstract}
\keywords{}
\subjclass[2020]{}
\maketitle
\tableofcontents

\section{Introduction}
 
Supersymmetry, a global symmetry
between bosons and fermions, provides invaluable insights to the non-perturbative aspects of general strongly coupled quantum field theories, and is deeply related to various areas of mathematics. This was elucidated by E. Witten in papers \cite{Witten:1982df} and  \cite{Witten:1982im}, written from physics and mathematics perspective.

Specifically, the Hilbert space $\cH=\cH^{+}\oplus\cH^{-}$ of  a supersymmetric quantum theory is graded by a fermion number and has an operator $(-1)^{F}$, which counts the number of fermions modulo two.
The precise non-perturbative information about the ground states of a supersymmetric quantum Hamiltonian\footnote{We use notation $\hat{H}$ to distinguish from the classical Hamiltonian $H$.} $\hat{H}$ is given by the Witten index
$$I=\mathrm{Str} \,e^{-\be \hat{H}}=\Tr (-1)^{F}e^{-\be \hat{H}},$$
which mathematically is related to the index of the Dirac operator associated with the supersymmetry algebra.

Remarkably, this interpretation led to a beautiful derivation of the Atiyah-Singer formula for the index of a Dirac operator using supersymmetric quantum mechanics \cite{Alvarez-Gaume:1983zxc}. In this context, the index of a Dirac operator is the Witten index, represented as a path integral for a certain supersymmetric sigma model. Since the index does not depend on the parameter $\be$ (playing the role of inverse temperature), in the high-temperature limit $\be\rightarrow 0$, the path integral localizes to the constant maps and becomes a finite-dimensional integral representing the $\hat{A}$\,-genus of a corresponding spin manifold. As it was eloquently explained by Atiyah \cite{atiyah1985circular}, this approach admits a natural equivariant cohomological interpretation as an infinite-dimensional version of the Duistermaat-Heckman formula for the loop space of a spin manifold, where the fixed points of the circle action are constant loops. This is a paradigmatic example of localization of a path integral, where the original infinite dimensional path integral reduces to a finite dimensional integral around the fixed points.

Subsequently, the idea of localization was extended to quantum field theories and led to many discoveries. To name a few, earlier important milestones of localizations were on topological and cohomological quantum field theories: computation of Donaldson invariant from a twisted four-dimensional supersymmetric gauge theory \cite{Witten:1988ze}, exact path integrals on 2d Yang-Mills theories \cite{Witten:1992xu} (see review \cite{Cordes:1994fc}), and the computation of Seiberg-Witten prepotential \cite{Nekrasov:2002qd} building upon \cite{ Losev:1997tp,Moore:1997dj}. More recently, the localization has been successfully applied to non-topological supersymmetric field theories, starting from four dimensional $N=2$ gauge theories \cite{Pestun:2007rz}, which then has been extended to three dimensions \cite{Kapustin:2009kz} and two dimensions \cite{Doroud:2012xw, Benini:2012ui,Benini:2013xpa}, see \cite{Pestun:2016zxk} for more comprehensive reviews. These discoveries confirm the common belief that localization of a supersymmetric path integral is always a computation of the supertrace and not of the trace.

Indeed, in all these examples the localization has been applied to `supersymmetric' observables, which are invariant under at least a single supercharge of the theory. However, it is clear that supersymmetric observables alone do not characterize the full physical theory. Therefore, a somewhat unorthodox question is the following: can one apply localization to the computation of a non-supersymmetric observable? In this paper, we focus on quantum mechanics and find that the answer is positive by extending the applicability of localization to a certain class of non-supersymmetric observables. Here we show that for some specific theories, it is possible to extend the localization principle to non-supersymmetric observables used to compute the partition function. 

Namely, a fundamental object of  a quantum theory is the partition function of a system,
$$Z(\beta)=\Tr e^{-\be \hat{H}},$$
whose dependence on the inverse temperature $\be$ is highly non-trivial. In the spectral geometry one studies $Z(\be)$ in case $\hat{H}=\frac{1}{2}\Delta$, the Laplace operator on a compact
Riemannian manifold $M$, acting on the Hilbert space $L^{2}(M)$ of square-integrable functions with respect to the Riemannian volume form on $M$. By definition,
$$Z(\beta)=\sum_{n=0}^{\infty}e^{-\be\lambda_{n}/2},$$
where $\lambda_{n}$ are the eigenvalues of $\Delta$. In general, only asymptotics of $Z(\be)=\Tr e^{-\frac{1}{2}\be\Delta}$ and of the corresponding heat kernel are known as $\be\to 0$, and no closed expression for $Z(\be)$ exists in geometric terms.

In special cases when $M$ is either a flat torus or a hyperbolic Riemann surface, the partition function $Z(\be)=\Tr e^{-\frac{1}{2}\be\Delta}$ can be computed exactly in terms of the underlying geometry of $M$. Thus in  case $M=S^{1}$ classical Jacobi inversion formula for the theta-series \cite{jacobi1828} represents $Z(\beta)$ as a sum over closed geodesics of the flat metric on $S^{1}$. In case $M$ is a hyperbolic compact Riemann surface, the corresponding result is the celebrated Selberg trace formula \cite{selberg1956harmonic}, which expresses $Z(\beta)$ as a sum over closed geodesics of the hyperbolic metric on $M$. These formulas can be thought of as a manifestation of the  principle ``Spectral trace = Matrix trace'' in the abelian and non-abelian settings.  An easier case is an explicit formula for the heat kernel of a Laplace operator for the bi-invariant metric on a compact semi-simple Lie group $G$, first obtained by L.D. Eskin \cite{MR0206535}, and later rediscovered and used by many authors \cite{Schulman:1968yv,dowker1970sum,Dowker:1970vu,marinov1979dynamics,Frenkel1984,Picken:1988ev,Camporesi:1990wm}.

One may ask when the matrix trace can be exactly computed by the path integral for a quantum particle on a Riemannian manifold $M$. This is obviously so in case of $S^{1}$, since
the path integral is Gaussian and can be computed exactly. In case when $M=G$, a compact simple Lie group with a bi-invariant metric, M.S. Marinov and M.V. Terentyev \cite{marinov1979dynamics} considered a semi-classical approximation of the path integral for a free quantum particle on $G$ and obtained the Eskin formula for the heat kernel, generalizing the observations of the exactness of semi-classical approximation in $SU(2)$ and $SU(N)$ cases by L. Schulman and J.S. Dowker \cite{Schulman:1968yv,Dowker:1970vu}.   It was indicated by R.F. Picken \cite{Picken:1988ev}  that this remarkable property of purely bosonic path integral on $G$ is related to the infinite-dimensional Duistermaat-Heckman formalism, outlined by Atiyah \cite{atiyah1985circular}.  However,  all these papers use an ad hoc addition of the so-called DeWitt term \cite{dewitt1957dynamical}, a `notorious' quantum correction to the Lagrangian for the path integral on curved spaces. Further, M. Gutzwiller \cite{gutzwiller1980classical} observed that when $M$ is a hyperbolic Riemann surface, the critical points of the path integral in the semi-classical approximation are closed geodesics; the resulting formula, called `Gutzwiller trace formula', is quite similar to the Selberg trace formula. However, Gutzwiller trace formula does not adequately reproduce the contribution of the identity element to the Selberg trace formula, and also uses an ad hoc addition of the DeWitt term.

Naturally, this calls for a question of whether there is a new localization principle that allows to compute pure bosonic partition function $Z(\be)$ by localizing some supersymmetric path integral to the closed geodesics. At first glance, this sounds rather counter-intuitive since, unlike the Witten index $I$,  $Z(\be)$ non-trivially depends on $\be$, so the standard localization principle does not apply.
Nevertheless, here we introduce a new principle of supersymmetric localization,  and use it to obtain trace formulas for $S^1$ and $G$ by localizing on closed geodesics.  
Note that in this examples we have $I=0$, and it is rather amusing that this vanishing of the Witten index provides a key for the answer! 

Namely, consider the supersymmetric Lagrangian $\mathcal{L}$ with the action $S$, in which fermion degrees of freedom totally decouple, so the Hamiltonian $H$ is purely bosonic.
Then the fermion part (assuming that it is contains only first time derivatives) has zero modes and the Witten index $I$, the supersymmetric path integral with periodic boundary conditions, is zero. We saturate the fermion
zero modes $\chi^{1},\dots,\chi^{n}$ by considering a following trace
$$I(\be)=\Tr\,\hat \chi^{1}\cdots \hat\chi^{n}(-1)^{F}e^{-\be \hat{H}},$$
which is non-zero.  If, in addition, the operator $\hat\chi^{1}\cdots\hat\chi^{n}$ is proportional to $(-1)^{F}$, then the index with zero modes insertion $I(\be)$ is equal (up to an overall numerical factor) to the partition function $Z(\beta)$! 
However, due to the presence of zero modes the functional $\chi^{1}\cdots\chi^{n}e^{-S_{E}}$, where $S_{E}$ is the Euclidean action, is no longer invariant under supersymmetry transformation, 
$\del(\chi^{1}\cdots\chi^{n}e^{-S_{E}})\neq 0$.  Thus the path integral of  $\chi^{1}\cdots\chi^{n}e^{-S_{E}+s\del V}$,  where $V$ is standard supersymmetric deformation associated with the circle action, is no longer $s$-independent (otherwise it would localize on constant loops).  However, in many cases there is another deformation $V$, containing higher derivatives of the fields,  such that  the path integral of $\chi^{1}\cdots\chi^{n}e^{-S+s V}$ is $s$-independent and in the limit $s\to\infty$ localizes on the closed geodesics!

The present paper was  greatly influenced by the work  of J.-M. Bismut in \cite{bismut2005hypoelliptic,bismut2008hypoelliptic,bismut2011hypoelliptic}, and it is illuminating to compare our approaches. Thus Bismut uses  hypoelliptic deformation on the cotangent bundle of the manifold that smoothly interpolates between the Laplacian and the geodesic flow; as it is clearly explained in \cite{bismut2005hypoelliptic}, such hypoelliptic deformation is a generalization of Witten's twist in Morse theory \cite{Witten:1982im}. Henceforth such approach can be thought of as a generalization of the Hamiltonian approach to supersymmetry, and one needs to find its physics interpretation. On the other hand, our approach is purely Lagrangian and extends the equivariant cohomology used for supersymmetric path integrals.

The content of the paper is the following. In Section \ref{subsec:slp} we recall standard basic facts on supersymmetric localization in the Hamiltonian  (using equivariant cohomology) and in the Lagrangian (using path integral) approaches. In Section \ref{new-susy-loc} we explicitly formulate our new localization principle. Namely, we start with elementary finite-dimensional Lemma  \ref{loc-fd} and carefully list all necessary conditions. They are based on the zero Witten index and assume that the system has fermion zero modes satisfying conditions (i)--(ii). Conditions (A)--(B) represent requirements on the deformation $V$, necessary for the Proposition \ref{new-loc} on the independence of the path integral of the parameter $s$. 

In Section \ref{sec:S1} we use Proposition \ref{new-loc}  to evaluate the partition function $Z(\be)$ on the circle $S^{1}$ with the flat metric using supersymmetric localization, which yields classical Jacobi
inversion formula for theta-series. In Section  \ref{sec:G} we derive the Eskin trace formula on the compact semi-simple Lie group $G$. Specifically, in Section \ref{SUSY-G} we present necessary facts about supersymmetric
particle on $G$, relegating standard details to Appendices  \ref{appsec:g}, \ref{appsec:symp} and \ref{free on G}. Finally, in Section \ref{E on G} we prove Theorem \ref{E-TF}, which gives a supersymmetric derivation of the Eskin trace formula on $G$.

\subsection{Acknowledgments} The first author (C.C.) thanks M. Dedushenko, Z. Komargodski, M. Mezei, and  M. Ro\v{c}ek for discussions and comments. C.C. is supported in part by the Simons Foundation grant 488657 (Simons Collaboration on the Non-Perturbative Bootstrap). The second author (L.T.) thanks J.-M. Bismut for the discussion at the conference ``Integrability, Anomalies and Quantum Field Theory'' in the IHES in February 2020, which stimulated the present paper.

\section{General Remarks on New Localization Principle} \label{sec:gen}

\subsection{Standard Localization Principle} \label{subsec:slp} We start by briefly recalling a finite-dimensional localization principle. Let $M$ be compact orientable $n$-dimensional manifold with an action  of the abelian group  $\mathrm{U}(1)=S^{1}$, and let $v$ the vector field corresponding to this action. The corresponding equivariant differential is
\beq\label{D-eq}
D = d- \iota_v
\eeq
where $\iota_{v}$ is the inner product operator with $v$. It satisfies
$$D^2=-\cL_v,$$
where $\cL_{v}$ stands for the Lie derivative, and is a differential in the subcomplex 
$\Omega^{\bullet}_{S^{1}}(M)$ of the complex $\Omega^{\bullet}(M)$, consisting of $S^{1}$-invariant differential forms on $M$; $ \al\in  \Omega^\bullet_{S^{1}}(M)$ if $\cL_v \al=0$.

Let $\al\in\Omega^{\bullet}(M)$ be an equivariantly closed form, $D\al=0$. The localization principle is the statement that for every $V\in\Omega^{1}_{S^{1}}(M)$ and $s\in\RR$ we have
\beq \label{eq:wi}
\int_M \al=\int_M \al \wedge e^{-s DV}.
\eeq
By a suitable choice of $V$ one can make $(DV)_{0}$, a component of $DV$ in $\Omega^{0}(M)$, positive semi-definite, so in the limit $s \rightarrow \infty$ the integral $\int_M \al$ localizes onto the zero loci of $(DV)_{0}$. In case when $M$ is even-dimensional and the circle action has only isolated fixed points, one gets Berline-Vergne localization formula \cite{berline2003heat}. 

The proof of \eqref{eq:wi} is very simple. We have
\begin{align*}
-\frac{d}{ds}\int_M \al \wedge e^{-s DV} & = \int_M \al \wedge DV\wedge e^{-s DV}\\
& = \int_M D(V\wedge \al \wedge e^{-s DV})  +\int_M V \wedge D(\al\wedge e^{-s DV})\\
&=0.
\end{align*}
Here the first integral in the second line is zero by the Stokes' theorem since the top component of an equivariantly exact form is exact, and the second integral is zero because of $D\al=0$ and $D^{2}V=0$.

The infinite-dimensional case was eloquently explained by Atiyah \cite{atiyah1985circular}, elaborating the observation by Witten. Namely, one replaces a finite-dimensional Riemannian manifold $M$ by its loop space $LM=\text{Map}(S^1_\be,M)$
where $S^{1}_{\beta}=\RR/\be\ZZ$ is a `thermal' circle. The loop space $LM$ is orientable when $M$ is a spin manifold and has a natural circle action with the vector field $v$. The equivariant differential on $LM$ has the same form \eqref{D-eq}, where now $d$ stands for the de Rham differential on $LM$. The cotangent bundle
$T^{{*}}LM$ carries a natural $S^{1}$-invariant $1$-form $\theta$,  dual to the vector field $v$ with respect to the Riemannian metric on $LM$. The closed $2$-form $\omega=d\theta$ plays the role of symplectic form on $LM$ (which is degenerate along closed geodesics), and the circle action on $LM$ is Hamiltonian: there is a function $H: LM\to\RR$ such that
$$i_{v}\omega=-dH.$$

Since $-H+\omega=D\theta$, the differential form $\al=e^{-H+\omega}\in\Omega^{\bullet}(LM)$ is equivariantly closed, so analogously to \eqref{eq:wi} for any $V\in\Omega^{1}_{S^{1}}LM$ we have 
\beq \label{eq:loc-loop}
\int_{LM}e^{-H+\omega}=\int_{LM} e^{-H+\omega -s DV}.
\eeq
Here integral over $LM$ stands for the integration of the top component of a differential form on the loop space. In particular, putting $V=\theta$, in the limit $s\to\infty$ we obtain
the Atiyah-Singer formula for the index of the Dirac operator on $M$.

This approach can be also formulated using supersymmetric quantum mechanics, as was originally proposed by Witten. 
Namely, consider a classical supersymmetric mechanical system with the Lagrangian $\LL$ and Hamiltonian $H$, having a single real supercharge $Q$ satisfying $\{Q,Q\}=2iH$, where $\{~,~\}$ is the graded Poisson bracket. After the quantization we get the simplest $N=1/2$ supersymmetric quantum system with real supercharge $ \hat Q$ satisfying $\hat Q^{2}=\hat H$, where quantum Hamiltonian $\hat{H}$ acts in the Hilbert space\footnote{Here it is assumed that we have even number of Majorana fermions.} $\cH$, naturally $\mathbb Z_2$-graded by a fermion number operator $(-1)^F$.

The Witten index is given by the path integral
\beq\label{Index-standard}
I=\Tr (-1)^{F} e^{-\be \hat H}=\bm{\int}e^{-S_{E}[x,\psi]}\cD x\cD\psi,
\eeq
where 
$$S_E[x,\psi]=\int_0^\be \LL_E(x,\dot{x};\psi,\dot\psi) d\tau$$ 
is the Euclidean action for the Lagrangian $\LL$, and $\cD x\cD\psi$ is a suitable supersymmetric path integration measure for the bosonic and fermionic degrees of freedom. The integration in \eqref{Index-standard} goes over for the periodic boundary conditions and we have
$$\del S_E=0\quad\text{and}\quad  \del (\cD x\cD\psi)=0,$$
where $\del$ is the Wick rotated (for Euclidean time) classical supersymmetry  transformation generated by a supercharge $Q$ --- the analogue of the equivariant differential $D$ in the loop space approach. Let $V[x,\psi]$ be an invariant  deformation, a functional of classical fields satisfying 
$$\del^{2}V=0.$$ 
Formally repeating the proof  of \eqref{eq:wi}, we obtain that for all $s$
\beq\label{PI-local-standar}
\bm{\int}e^{-S_{E}}\cD x\cD\psi=\bm{\int}e^{-S_{E}-s\del V}\cD x\cD\psi.
\eeq
In case $S_{E}=\del V$ one can take $V=S_E$, so in the limit $s\to \infty$ the path integral localizes on the zero locus of $S_E$. The latter is nothing but the set of constant loops, arising from the standard kinetic term in the action. Finally, the relation $V=\delta\theta$ and the identification of $\delta$ with $D$ establishes equivalence with the loop space approach. 

\subsection{New Supersymmetric Localization Principle}\label{new-susy-loc}
We start by considering the following simple finite-dimensional example. 

Let $M$ be compact orientable $n$-dimensional manifold with circle action and let $N=M\times S^{1}$. The $n+1$-dimensional manifold $N$ has an obvious circle action and let $\alpha$ be an equivariant closed form on $N$ such that $i_{u}\al=0$, where $u$ is a vector filed on $N$ which at a point $(p,\vp)\in N$ is $d/d\vp$, where $\vp\!\!\mod 2\pi$ is the coordinate on $S^{1}$. In other words, $\al$ does not contain the $1$-form $d\vp$ as a factor and the top component of $\al$ has degree $n$. The standard localization principle applies to the integral $\int_{N}\al$, which is obviously $0$, but does not apply to the integral 
$\int_{N}d\vp\wedge \al$, since the form $d\vp\wedge \al$ is not equivariantly closed. However, we have the following statement.
\begin{lemma} \label{loc-fd} Let $V\in\Omega^{1}_{S^{1}}(N)$ be such that $i_{u}V=0$ and $i_{u}DV=0$. Then for all $s$
$$\int_{N} d\vp\wedge\al=\int_{N}d\vp\wedge\al\wedge e^{-s DV}.$$
\end{lemma}
\begin{proof}
It follows the derivation of \eqref{eq:wi}. Namely,
\begin{align*}
-\frac{d}{ds}\int_N d\vp\wedge\al\wedge e^{-s DV} & = \int_N d\vp\wedge  \al \wedge DV\wedge e^{-s DV}\\
& = \int_{N} D(V\wedge d\vp\wedge\al \wedge e^{-s DV})  +\int_N V\wedge D( d\vp\wedge\al\wedge e^{-s DV})\\
&=\int_N V\wedge D( d\vp)\wedge\al\wedge e^{-s DV}- \int_N V\wedge d\vp\wedge D(\al\wedge e^{-s DV})\\
&=0.
\end{align*}
As before, the first integral in the second line is zero by the Stokes' theorem, and the last integral in the third line is zero because $D\al=0$ and $D^{2}V=0$. The new feature is first integral in the third line,
which is zero because $D(d\vp)=-1$ and conditions $i_{u}V=0$ and $i_{u}DV=0$ imply that the form $V\wedge\al\wedge e^{-s DV}$ does not contain $d\vp$, so its top component is zero.
\end{proof}

As in the previous section, in the infinite-dimensional case, we start with a classical supersymmetric mechanical system with the Lagrangian $\LL$ and Hamiltonian $H$, having a single real supercharge $Q$ satisfying $\{Q,Q\}=2iH$, where $\{~,~\}$ is the graded Poisson bracket. As a basic example, consider a supersymmetric particle on the $n$-dimensional Riemannian manifold $M$ with the Lagrangian\footnote{Here we use the `physical' time and imaginary $i$ is needed to make quadratic expressions in $\psi$ ``real'', that is, 
$\overline{i\psi^{\mu}\psi^{\nu}}=-i\psi^{\nu}\psi^{\mu}=i\psi^{\mu}\psi^{\nu}$.}

\begin{align} \label{L-particle}
\LL & =\frac{1}{2}\la\dot{x},\dot{x}\ra+\frac{i}{2}\la \psi,\nabla_{\dot{x}}\psi\ra.
\end{align}
Here $x(t)$ is a path in $M$ expressed in terms of local coordinates $x=(x^{1},\dots, x^{n})$, $\dot{x}$ is the velocity vector field along the path, $\la~,~\ra$ is the Riemannian inner product in $T_{x(t)}M$ and $\psi(t)=(\psi^{1}(t),\dots,\psi^{n}(t))\in\Pi T_{x(t)}M$ is a tangent vector field along the path with the reverse parity (i.e, with the anti-commuting values). Finally, $\nabla$ in \eqref{L-particle} is an arbitrary metric connection in $TM$ with totally anti-symmetric torsion. The latter condition guarantees the standard supersymmetry
\beq\label{susy-standard}
\delta x^{\mu}=i\psi^{\mu},\quad \del \psi^{\mu}=-\dot{x}^{\mu},\quad\mu=1,\dots,n,
\eeq 
with a real supercharge $Q$, obtained by the Noether theorem (see Sections \ref{sec:S1} and \ref{sec:G}).

\begin{remark} \label{even-odd} In case of even $n$, after the quantization we get $N=1/2$ supersymmetric quantum particle on $M$ with a single real supercharge $ \hat Q$ and quantum Hamiltonian $\hat{H}=\frac{1}{2}\hat{Q}^{2}$ acting in the $\mathbb Z_2$-graded Hilbert space $\cH$ with the fermion number operator 
$$(-1)^F=c_n 2^{n/2}\hat\psi^{1}\cdots\hat\psi^{n}.$$
Here  $\hat\psi^{\mu}$ are Majorana (Hermitian) fermion operators satisfying canonical anti-commutation relations
\beq\label{cac}
[\hat\psi^{\mu},\hat\psi^{\nu}]=\delta^{\mu\nu}\hat I,
\eeq
where $[~,~]$ is a graded commutator and $\hat I$ is an identity operator acting on the  fermion Hilbert space. The overall constant factor $c_n$ has been introduced to ensure that fermion number operator is Hermitian and satisfies the identity 
\beq \label{F-square}
\left\{(-1)^F\right\}^2=\hat{I}.
\eeq
Namely,  it follows from \eqref{cac} that we have two choices $c_n=\pm i^{n(n-1)/2}$, which simply reflects the $\mathbb Z_2$ ambiguity in the definition of the total fermion parity. The choice of $c_{n}$ is reflected in the overall sign factor in the fermionic path integral measure.

In the case of odd $n$  it is well-known (see \cite{Delmastro:2021xox} for the recent discussion and the references therein) that Majorana fermions are anomalous in the sense that there is no Hilbert space interpretation.\footnote{Nevertheless, such system with an odd number of Majorana fermions are still physically well-defined in terms of path integral formulation.} To remedy this situation, we replace
$M$ by $M\times S^{1}$ and trivially add a single term $\psi^{n+1}\dot\psi^{n+1}$ to the Lagrangian, while keeping the kinetic term unchanged. Thus in this case we also have a graded Hilbert space $\cH$ where $\hat{H}$ acts, and it what follows we will always assume such modification for odd $n$.
\end{remark}
Like in the finite-dimensional example, our new localization principle applies to the systems with zero Witten index,
$$I=\Tr (-1)^{F} e^{-\be \hat H}=\bm{\int}e^{-S_{E}[x,\psi]}\cD x\cD\psi=0.$$
This situation is quite general: the index of the standard Dirac operator vanishes on arbitrary odd-dimensional spin manifolds, on compact Lie groups, etc. Specifically, we assume the following conditions.
\begin{enumerate}
\item [(i)] \label{1} The vanishing of the index is solely due to the presence of certain fermionic zero modes $\chi_1,\dots,\chi_n$ in the path integral.\footnote{In some cases, when one starts with a supersymmetric system with a vanishing $I$ but without fermionic zero modes, the robustness of the index allows one to smoothly deform the system for which $I$ vanishes exclusively due to the presence of fermionic zero modes.}
\item[(ii)] \label{2} In the Hilbert space $\cH$, the zero modes can be saturated by the `local' operators $\hat \chi_\mu$ 
%\footnote{\label{ft:zero} It is important to remark that the role of $\hat \chi_n$ is just to saturate the zero modes $\chi_n$ in the path integral, and they are `not' an exact dual description of the same physical object. This is because zero modes in the path integral in general have no exact operator counter part in the Hamiltonian picture since they are non-local in spacetime generically, especially w.r.t. any foliation of spacetime associated to the Hamiltonian formulation. 
that satisfy canonical anti-commutation relations \eqref{cac} and the relation
\beq\label{chi-F}
c_n\hat\chi_{1}\cdots\hat\chi_{n}=2^{-n/2}(-1)^{F},
\eeq
where $n=\dim M$ if $M$ is even-dimensional, and $n=\dim M+1$ otherwise.
\end{enumerate}
It follows from (i)--(ii) that\footnote{\label{ft:measure}The factor $c_n$ can be absorbed into the definition of the measure $\mathscr D\psi$ to make it consistent with the trace side on the Hilbert space. See section \ref{sec:G} for relevant details on the measure.  }
\beq\label{Tr=Str}
2^{-n/2}\Tr e^{-\be\hat{H}}=\mathrm{Str}\,c_n \hat\chi_{1}\cdots\hat\chi_{n}e^{-\be \hat{H}}=\bm{\int}  \chi_1\cdots \chi_n e^{-S_{E}[x,\psi]}\cD x\cD\psi,
\eeq
 so the path integral with periodic boundary conditions for fermions, with insertion of the zero modes, computes the trace in the Hilbert space $\cH$ and not the supertrace! 
 
\begin{remark} \label{hat-chi} We emphasize that the insertion of $\chi_\mu$ saturates the zero modes in the path integral  side of \eqref{Tr=Str} and makes it non-zero, while the insertion of local operators $\hat \chi_\mu$ saturates the fermion number operator $(-1)^{F}$ in the spectral side of \eqref{Tr=Str}. It is important to observe that the operators $\hat \chi_\mu$ are not naive quantization of the zero modes, since the latter are generically non-local in time and have no counterpart in the 
Hamiltonian formalism. The local operators $\hat\chi_{\mu}$ are just $\hat\psi^{\mu}(0)$, where fermion operators are normalized to satisfy \eqref{cac}.

%In Section \ref{E on G} we will see that the fermion zero modes are
%$$\chi^{a}=\frac{1}{\beta}\int_{0}^{\beta}\psi^{a}(\tau)d\tau,$$
%and corresponding fermion operators $\hat\chi^{a}$ in the Hamiltonian formalism  are obtained by normalizing the operators $\hat\psi^{a}(0)$.
\end{remark}
\begin{remark}\label{zero-modes-general} It should also be noted that if the fermion zero modes are present, insertions of $\chi^{a}$ and of $\psi^{a}(0)$ into the path integral give the same result. However, in case when there are no  fermion zero modes, the correct matching of the spectral and path integral sides in \eqref{Tr=Str} is provided by the insertion of $\psi^{a}(0)$.
\end{remark}
%n general have no exact operator counter part in the Hamiltonian picture since they are non-local in spacetime generically, especially w.r.t. any foliation of spacetime associated to the Hamiltonian formulation. 

%It is important to observe that the role of the operators $\hat \chi_n$ is to saturate the fermion number operator is just to saturate the zero modes $\chi_n$ in the path integral, and they are `not' an exact dual description of the same physical object. This is because zero modes in the path integral in general have no exact operator counter part in the Hamiltonian picture since they are non-local in spacetime generically, especially w.r.t. any foliation of spacetime associated to the Hamiltonian formulation. 

%Note that even though $\hat \chi_n$ is not an exact dual of $\chi_n$ (see footnote \eqref{ft:zero}), the fact that path integral vanishes in the presence of any unsaturated fermionic zero mode guarantee the 2nd equality of \eqref{Tr=Str}.
%\end{remark}
Though the formula \eqref{Tr=Str} looks very appealing, the standard supersymmetric localization discussed in Section \ref{subsec:slp} no longer applies, because
$$\del (\chi_1\cdots \chi_ne^{-S_E})\neq 0,$$
and the path integral in \eqref{Tr=Str} non-trivially depends on $\be$.

To formulate our new localization principle, we further assume that fermion zero modes satisfy the following conditions
\beq\label{zm-ortog}
\int\delta\chi_{\mu}d\chi_{\mu}=0,\quad \mu=1,\dots,n,
\eeq
which mean that $\delta\chi_{\mu}$ does not contain fermion degree freedom $\chi_{\mu}$, which we will indicate as  $\delta\chi_{\mu}\perp\chi_{\mu}$.  The following definition is fundamental.
\begin{definition} \label{def-non-standard} Invariant deformation for a  supersymmetric system with fermion zero modes $\chi_{1},\dots,\chi_{n}$ satisfying \eqref{zm-ortog} is a functional $V$ with a Grassmannian odd parity, satisfying the following conditions.
\begin{itemize}
\item[(A)] \label{A} $V$ is invariant,
$$\delta^{2}V=0.$$
\item[(B)] \label{B} $V, \del V\perp\chi_{\mu}$, i.e.
$$\int Vd\chi_{\mu}=\int\delta V d\chi_{\mu}=0,\quad \mu=1,\dots,n.$$
\end{itemize} 
\end{definition}
Note that condition (A) is standard, while condition (B), the absence of fermion zero modes in $V$ and $\del V$, is a completely new requirement. It is rather constraining and forces $V$ to explicitly depend on the first time derivatives of fermion degrees of freedom. 
The analog of Lemma \ref{loc-fd} is the following surprisingly simple statement.
\begin{proposition} \label{new-loc} Let $S_{E}$ be the Euclidean action of the supersymmetric theory with fermion zero modes $\chi_{1},\dots, \chi_{n}$ satisfying \eqref{zm-ortog}. Then for all $s$ we have
$$\bm{\int}\chi_{1}\cdots\chi_{n}e^{-S_{E}}\cD x\cD\psi=\bm{\int}\chi_{1}\cdots \chi_{n}e^{-S_{E}-s\del V}\cD x\cD\psi,$$
where $V$ is a deformation satisfying conditions (A)--(B).
\end{proposition}
\begin{proof} We have 
\begin{gather*}
-\frac{d}{ds}\bm{\int}\chi_{1}\cdots \chi_{n}e^{-S_{E}-s\del V}\cD x\cD\psi
=\bm{\int}\chi_{1}\cdots \chi_{n}\del V e^{-S_{E}-s\del V}\cD x\cD\psi\\
=(-1)^{n}\bm{\int}\left[\del(\chi_{1}\cdots\chi_{n}V) -\del(\chi_{1}\cdots \chi_{n})V\right]e^{-S_{E}-s\del V}\cD x\cD\psi \\
=0.
\end{gather*}
Here the integral of the first term in the second line is zero by Stokes' theorem, and the integral of the second term is zero since by conditions \eqref{zm-ortog} and (A)--(B) it does not contain zero modes.
\end{proof}
As was remarked before, conditions (A)--(B) imply that localization $V$ contains higher derivatives and when $\del V$ is positive semi-definite, the path integral localizes to the solutions of the second or higher order ordinary differential equations. In particular, such $V$ provides a mechanism for localizing on the geodesic locus. We will consider this application of the new localization principle in Sections \ref{sec:S1} and \ref{sec:G}.

Note that Proposition \ref{new-loc} is rather general and does not use condition (ii). Actually, one can relax the condition (ii) and still obtain the information about the partition function when the corresponding supersymmetric system has decoupled bosonic and fermionic degrees of freedom, so
$$\cH=\cH_{B}\otimes\cH_{F}\quad\text{and}\quad \hat{H}=\hat{H}_B\otimes I_F+ I_B\otimes \hat{H}_F,$$
where the fermion Hilbert space $\cH_{F}$ is irreducible $2^{n/2}$-dimensional
Clifford module (see Remark \ref{even-odd}).
In this case, the supertrace in \eqref{Tr=Str} is proportional to bosonic trace, 
$$\text{Str}\,c_{n}\hat\chi_{1}\cdots\hat\chi_{n}e^{-\be \hat{H}}=2^{-n/2}\Tr e^{-\be \hat{H}_{B}}\otimes e^{-\be \hat{H}_{F}}=2^{-n/2}\Tr_{\cH_{F}}e^{-\be \hat{H}_{F}}\cdot\Tr_{\cH_{B}}e^{-\be\hat{H}}.$$ 
In particular, if the supersymmetric system has free fermion part, then quantum Hamiltonian $\hat{H}$ acts as identity operator on the fermion Hilbert space $\cH_{F}$ and
$$\text{Str}\,c_{n}\hat\chi_{1}\cdots\hat\chi_{n}e^{-\be\hat{H}}=\Tr_{\cH_{B}}e^{-\be\hat{H}}.$$  
Using the representation \eqref{Tr=Str} and Proposition \ref{new-loc}, we obtain a pure bosonic trace formula by localizing the supersymmetric path integral \eqref{Tr=Str} in the limit $s\to\infty$. 

We will illustrate this general procedure in the next two sections by considering a simple case $M=S^{1}$ with the flat metric, and $M=G$, a compact connected semi-simple Lie group with the bi-invariant metric. In the first 
case the Lagrangian is quadratic so fermion degrees of freedom trivially decouple, while in the second case instead of the Levi-Civita connection we use invariant flat connection on $G$ and corresponding Kostant cubic Dirac operator.

\section{Localization Proof of the Trace formula on $S^1$}  \label{sec:S1}
In this section, we consider the simplest application of the localization principle explained in section \ref{sec:gen} by providing localization proof of the trace formula on $S^1$. This identity, the famous Jacobi inversion formula\footnote{It is a special case of general formula for Jacobi theta-function \cite{jacobi1828} and was already known to Poisson \cite{poisson1827}.}, is related to the spectrum of the Laplace operator $\Delta=-\dfrac{d^{2}}{dx^{2}}$ on $\cH_{B}=L^{2}(S^{1})$, where $S^1=\mathbb R/2\pi\mathbb Z$. Namely, consider a free quantum particle of mass $1$ on $S^{1}$ with the Hamiltonian $\hat{H}=\frac{1}{2}P^{2}=\frac{1}{2}\Delta$, and let $Z(\be)=\Tr_{\cH_{B}} e^{-\be\hat{H}}$ be the corresponding thermal partition function. We have
\beq \label{eq:trS1}
Z(\be)=\sum_{n=-\infty}^{\infty}e^{-\be  n^2/2}=\sqrt{{2\pi \ov\be}}\sum_{n=-\infty}^{\infty}e^{-2\pi^2n^2/\be},\quad \be> 0.
\eeq
\begin{remark} In the number theory, Jacobi inversion formula is usually written as
$$\theta\left(-\frac{1}{z}\right)=\sqrt{\frac{z}{i}}\,\theta(z),\quad\text{where}\quad\theta(z)=\sum_{n=-\infty}^{\infty}e^{\pi in^{2}z}\quad\text{and}\quad \im z>0$$
is the Jacobi theta series; the branch of the square root is defined by $\sqrt{1}=1$. Of course, Jacobi inversion formula can be easily obtained from the Poisson summation formula
\begin{equation*}
\sum_{n=-\infty}^{\infty}f(n)=\sum_{n=-\infty}^{\infty}\hat{f}(n),
\end{equation*}
where 
$$\hat{f}(k) =\int_{-\infty}^{\infty}f(x)e^{2\pi i kx}dx.$$
Conversely, multiplying the Jacobi inversion formula for $z=it$ by a `nice' function $g(t)$, integrating over $t$ and denoting $$f(x)=\int_{0}^{\infty}e^{-\pi tx^{2}}g(t)dt,$$
we get the Poisson summation formula.
\end{remark}

It is rather remarkable that one can get Jacobi inversion formula \eqref{eq:trS1} as an elementary application of the new supersymmetric localization principle, formulated in Section \ref{new-susy-loc}. Namely, consider the simplest
supersymmetric Lagrangian
\beq\label{L-circle}
\LL={1\ov 2}(\dot x^2+i\psi \dot \psi)
\eeq
which has a standard $N=1/2$ supersymmetry \eqref{susy-standard}. We have
\begin{align*}
\delta \LL &=\frac{1}{2}(2\delta\dot{x}\,\dot{x}+i\delta\psi\dot\psi-i\psi\delta\dot\psi)=\frac{i}{2}\frac{d}{dt}(\dot{x}\psi),
\end{align*}
so by the Noether theorem, 
\begin{align*}
iQ &= \delta x\frac{\pa \LL}{\pa\dot{x}}+\delta\psi\frac{\pa \LL}{\pa\dot\psi}-\frac{i}{2}(\dot{x}\psi)=i\dot{x}\psi,
\end{align*}
 where $\dfrac{\pa}{\pa \dot{\psi}}$ stands for the left partial derivative in Grassmann variable.
Using the graded Poisson bracket 
$$\{\psi,\psi\}=i$$
for Majorana fermion, we get the classical Hamiltonian 
\beq\label{Q-H-0}	
H=-\frac{i}{2}\{Q,Q\}=\frac{1}{2}p^{2},
\eeq
where $p=\dot{x}$ by the Legendre transform.

The corresponding Witten index $I$ is given by \eqref{Index-standard}, where the Euclidean action is
$$S_{E}=\int_{0}^{\be}{1\ov 2}(\dot x^2+\psi \dot \psi)d\tau,$$
where now dot stands for the derivative to the Euclidean time $\tau$;  due to the presence of the zero mode
$$\chi=\frac{1}{\be}\int_0^\be \psi\, d\tau $$
we have $I=0$.

We can saturate the fermionic zero mode by inserting $\hat \chi=\hat \psi$ in the trace side. Since after the quantization we have $\hat\chi^{2}=\frac{1}{2} I$, conditions (i)--(ii) of our localization principle are satisfied (here we are using Remark \ref{even-odd}).
Moreover, $\hat{H}=\frac{1}{2}\hat{Q}^{2}=\frac{1}{2}\Delta$, so the fermion degrees of freedom totally decouple and we have
\beq
\text{Str}\,\hat \chi e^{-\be  \hat H}=\Tr_{\cH_{B}} e^{-\be \hat H}=Z(\be).
\eeq

On the path integral side, we have
\beq \label{I-circle}
I(\chi,\be)=\bm{\int}_{ \Pi TLS^{1}}\chi e^{-S_{E}}\cD x\cD\psi.
\eeq
Here $e^{ix(\tau)}\in LS^{1}=\mathrm{Map}(S^{1}_{\be}, S^{1})$ is the loop group of the circle $S^{1}$, where $S^{1}_{\be}=\RR/\be\ZZ$ the `thermal' circle , $\psi(\tau)\in\Pi L\RR$ are $\beta$-periodic functions with anticommuting values, and integration goes over $ \Pi TLS^{1}=LS^{1}\times\Pi L\RR$, the tangent bundle of $LS^{1}$ with the reverse parity of the fibers, and
$\cD x\cD\psi$ denotes the path integral `measure' in $ \Pi TLS^{1}$, defined as follows.

Lett $u_{0}(\tau)=\dfrac{1}{\sqrt\be}$, $u_{n}(\tau)=\sqrt{\dfrac{2}{\be}}\sin\omega_{n}\tau$ for $n>0$ and $u_{n}(\tau)=\sqrt{\dfrac{2}{\be}}\cos\omega_{n}\tau$ for $n<0$, where $\omega_{n}=\dfrac{2\pi n}{\be}$,  be the orthonormal eigenfunctions of the operator $-\dfrac{d^{2}}{d\tau^{2}}$ in the real Hilbert space $L^{2}(S^{1}_{\be};\RR)$.
Consider the eigenfunction expansions
\begin{align*}
x(\tau) =\sum_{n=-\infty}^{\infty}c_{n}u_{n}(\tau)\quad\text{and}\quad\psi(\tau) =\sum_{n=-\infty}^{\infty}\psi_{n}u_{n}(\tau),
\end{align*}
where\footnote{This reflects the fact that $x(\tau)$ are real-valued and $\psi(\tau)$ is Majorana fermion.} $\bar{c}_{n}=c_{n}$, $\bar{\psi}_{n}=\psi_{n}$, 
 and put
 \begin{equation}\label{measure-2}
\cD x\cD\psi=dc_{0}\prod_{n=1}^{\infty}dc_{n}dc_{-n}\, d\psi_{0}\prod_{n=1}^{\infty}d\psi_{-n}d\psi_{n}.
\end{equation}
As the result, bosonic and fermionic Gaussian path integrals with respect to $\cD x\cD\psi$ are expressed in terms of the regularized determinants  and Pfaffians of the corresponding differential operators in the real Hilbert space $L^{2}(S^{1}_{\be};\RR)$. In particular, since the eigenvalues of $-\dfrac{d^{2}}{d\tau^{2}}$ are $\omega^{2}_{n}$,  
we have by the standard zeta function regularization
\begin{equation}\label{det-real}
\det(-\pa^{2}_{\tau})=\prod_{n=1}^{\infty}\omega_{n}^{4}=\beta^2,
\end{equation}
which coincides with $\det(-\pa^{2}_{\tau})$  in the complex Hilbert space $L^{2}(S^{1}_{\be})$. For the skew-symmetric operator $\pa_{\tau}$ in $L^{2}(S^{1}_{\be};\RR)$ we have 
$$
\pa_{\tau}\begin{pmatrix} u_{n}\\u_{-n}\end{pmatrix}=\begin{pmatrix} 0 &\omega_{n}\\ -\omega_{n} & 0\end{pmatrix}\begin{pmatrix} u_{n}\\u_{-n}\end{pmatrix},\quad n>0,$$
so
\begin{equation}\label{pfaf-real}
\mathrm{Pf}(\pa^{2l+1}_{\tau})=\prod_{n=1}^{\infty}\omega_{n}^{2l+1}=\beta^{l+\frac{1}{2}}.
\end{equation}
\begin{remark}\label{rescaling}
The bosonic and fermionic measures $\cD x$ and $\cD\psi$ have remarkable scaling property: they are invariant under the change of variables $c_{0}\mapsto ac_{0}$ and $c_{n}\mapsto ac_{n}$ for any $a>0$, and $\psi_{0}\mapsto b\psi_{0}$ and $\psi_{n}\mapsto b\psi_{n}$ for any $b>0$. Indeed, the Jacobian for the bosonic change of variables is
$$a\prod_{n=1}^{\infty}a^{2}=a \,e^{2\zeta(0)\log a}=1,$$
since for the Riemann zeta-function $\zeta(0)=-1/2$, 
and similarly for the fermionic variables.
We  choose $a=b=1/\sqrt{\be}$, so that $ac_{0}$ varies from $0$ to $2\pi=\mathrm{Vol}(S^{1})$, and $\chi=a\psi_{0}$ satisfies condition (ii) in Section \ref{new-susy-loc}.
\end{remark}

It follows from the discussion in Section  \ref{new-susy-loc}. that the simplest invariant deformation of $I(\chi,\be)$ is given by the following higher-derivative functional\footnote{It is noteworthy that such higher-derivative term resembles the one in the linearized Schwarzian action as in \cite{Stanford:2017thb}. }
\beq \label{V-S1}
V=-{1\ov 2}\int_0^\be \dot \psi\ddot x\, d\tau, 
\eeq
so
\beq \label{del-V-S1}
\delta V={1\ov 2}\int_0^\be (\ddot x^2 +\dot \psi \ddot \psi)d\tau,
\eeq
where we are using Euclidean version of formulas \eqref{susy-standard},
\beq\label{susy-standard-Euclid}
\del x=\psi\quad\text{and}\quad \del\psi=-\dot{x}.
\eeq

Indeed, it is elementary to verify that $\del\chi \perp \chi$ and $V,\del V\perp\chi$. Thus
\beq
I(\chi,\be)=\bm{\int}_{ \Pi TLS^{1}}\chi e^{-S_{E}-s\del V}\cD x\cD\psi
\eeq
and according to \eqref{del-V-S1}, in the limit $s\to\infty$  the path integral localizes on closed geodesics satisfying $\ddot x=0$, i.e., $x(\tau)=c+\omega_{n}\tau$.

Let $\Omega S^{1}$ be the space of based loops,
$$\cD^{\prime} x\cD'\psi=\prod_{n=1}^{\infty}dc_{n}dc_{-n}\,\prod_{n=1}^{\infty}d\psi_{-n}d\psi_{n}$$
be the corresponding path integral measure on $\Pi T\Omega S^{1}$, and $x_{cl}(\tau)=\omega_{n}\tau$ be isolated geodesics in $\Omega S^{1}$.
Using functional analog of the elementary formula
$$\delta(f(x))=\sum_{f(x_{l})=0}\frac{\delta(x-x_{l})}{|f'(x_{l})|},$$
 translation invariance of the action, $\mathrm{Vol}(S^{1})=2\pi$, and formulas \eqref{det-real}--\eqref{pfaf-real}, we easily compute
\begin{align*}
I(\chi,\be)&=2\pi \lim_{s\to\infty}\bm{\int}_{ \Pi T\Omega S^{1}} e^{-S_{E}-s\del V}\cD' x\cD'\psi\\ 
&=2\pi \cdot (2\pi)^{\zeta(0)} \bm{\int}_{ \Pi T\Omega S^{1}} e^{-S_{E}[x,\psi]}\delta(\ddot{x})\delta(\psi)\mathrm{Pf}(-\pa^{3}_{\tau})\cD' x\cD'\psi\\ 
&=2\pi\cdot (2\pi)^{\zeta(0)} \bm{\int}_{\Omega S^{1}} e^{-S_{E}[x,0]}\sum_{x_{cl}}\frac{\delta(x-x_{cl})}{|\det(\pa^{2}_{\tau})|}\mathrm{Pf}(-\pa^{3}_{\tau})\cD' x\\
&=2\pi\cdot (2\pi)^{\zeta(0)}\sum_{x_{cl}}e^{-\frac{1}{2}\int_{0}^{\be}\dot{x}_{cl}^{2}d\tau}\frac{\mathrm{Pf}(-\pa^{3}_{\tau})}{\det(-\pa^{2}_{\tau})}\\
&=\sqrt{2\pi \ov \be}\sum_{n=-\infty}^{\infty}e^{-2\pi^2 n^2/\be}.
\end{align*}

\begin{remark} One can think of $2\zeta(0)=2\sum_{n=1}^{\infty}1=-1$ as the regularized dimension of $\Omega S^{1}$, so the prefactor $(2\pi)^{\zeta(0)}$ is analogous to the prefactor $(2\pi)^{\dim M/2}$ in the finite-dimensional Gaussian integration. 
\end{remark}

\section{Localization Proof of the Eskin Trace formula on $G$}\label{sec:G}
\subsection{Supersymmetric Particle on $G$} \label{SUSY-G} Here we describe classical and quantum aspects of the supersymmetric particle on a compact semi-simple Lie group $G$. We summarized some basic facts on $G$ in the appendix \ref{a-2}. As in Section \ref{new-susy-loc}, let $x^{\mu}(t)$ be local coordinates in the neighborhood of $g(t)$ for a path $g(t)$ in $G$, and let $\psi(t)\in\Pi T_{g(t)}G$ be a vector field along the path with the reverse parity. We consider the following Lagrangian $\LL$ of a free supersymmetric particle on $G$ endowed with the connection with totally anti-symmetric parallelizing torsion, 
\beq\label{L-CV-0}
\LL=\frac{1}{2}(\dot{x},\dot{x}) +\frac{i}{2}( \psi,\nabla^{-}_{\dot{x}}\psi),
\eeq
where $\nabla^{-}$ is a flat connection on $G$\footnote{Equivalently, we could start with the choice of the right-invariant flat connection $\nabla^+$.}, defined by the condition that the elements of $\frak{g}$ are parallel vector fields with the parallel transport given by the left translations $(L_{g})_{*}$. In terms of the local coordinates the Lagrangian takes the form
\beq\label{L-CV-1}
\LL=\frac{1}{2}g_{\mu\nu}\dot{x}^{\mu}\dot{x}^{\nu}+\frac{i}{2}g_{\mu\nu}(\psi^{\mu}\dot\psi^{\nu}+\Gamma^{\nu}_{\lambda\rho}\dot{x}^{\lambda}\psi^{\mu}\psi^{\rho}),
\eeq
where Christoffel symbols are given by 
 $\Gamma^{\nu}_{\rho\lambda}$ are \beq\label{C-symbol}
\Gamma^{\nu}_{\lambda\mu}=\theta^{\nu}_{a}\pa_{\lambda}\theta^{a}_{\mu}
\eeq
(see appendix \ref{a-2} for notations).
The supersymmetry transformation\footnote{Here we use the Lorentzian time and imaginary $i$ is needed to make quadratic expressions in $\psi$ ``real'', that is, 
$\overline{i\psi^{\mu}\psi^{\nu}}=-i\psi^{\nu}\psi^{\mu}=i\psi^{\mu}\psi^{\nu}$.}  is given by standard formulas \eqref{susy-standard},
\beq\label{susy}
\delta x^{\mu}=i\psi^{\mu},\quad \del \psi^{\mu}=-\dot{x}^{\mu},\quad\mu=1,\dots,n.
\eeq 

In terms of the current $J=g^{-1}\dot{g}=J^{a}T_{a}\in\frak{g}$ and $\psi=\psi^{a}T_{a}\in\Pi\frak{g}$, where $\psi^{a}=\theta^{a}_{\mu}\psi^{\mu}$, the supersymmetry transformation takes the following form 
\begin{align}
\delta g & =ig\psi,\label{g-susy}\\
\delta \psi &=-J-i\psi\psi, 
\label{psi-susy}\\
\delta J &=i(\pa_t +\text{ad}_J)\psi,\label{J-susy}
\end{align}
where the formula for $\del J$ is a direct consequence of the first two. In the formula \eqref{g-susy} it is understood that $g\psi=(L_{g})_{*}\psi$, and in \eqref{psi-susy}
$$\psi\psi=\frac{1}{2}\psi^{a}\psi^{b}(T_{a}T_{b}-T_{b}T_{a})=\frac{1}{2}[\psi,\psi].$$

Indeed, formula \eqref{g-susy} immediately follows from the first formula in \eqref{susy}, while for \eqref{psi-susy} we have the following simple computation, using \eqref{M-C}:
\begin{align*}
\delta \psi^{a} &=-\theta^{a}_{\mu}\dot{x}^{\mu}+i\pa_{\nu}\theta^{a}_{\mu}\psi^{\nu}\psi^{\mu}
=-J^{a}+\frac{i}{2}(\pa_{\nu}\theta^{a}_{\mu}-\pa_{\mu}\theta^{a}_{\nu})\psi^{\nu}\psi^{\mu}\\
&=-J^{a}-\frac{i}{2}f^{a}_{bc}\theta^{b}_{\nu}\theta^{c}_{\mu}\psi^{\nu}\psi^{\mu}
=-J^{a}-\frac{i}{2}f^{a}_{bc}\psi^{b}\psi^{c}\\
&=-J^{a}-i(\psi\psi)^{a}.
\end{align*}

In the Hamiltonian formalism, Legendre transform associated with the Lagrangian \eqref{L-CV-0} formally gives\footnote{Here $\dfrac{\pa}{\pa \dot{\psi}^{\mu}}$ stands for the left partial derivative in Grassmann variables.} 
\beq \label{Leg-transform}
\tilde{p}_{\mu}=\frac{\pa \LL}{\pa \dot{x}^{\mu}}=g_{\mu\nu}\dot{x}^{\nu}+\frac{i}{2}g_{\rho\nu}\Gamma^{\nu}_{\lambda\mu}\psi^{\rho}\psi^{\lambda} \quad\text{and}\quad \tilde\pi_{\mu}=\frac{\pa \LL}{\pa\dot{\psi}^{\mu}}= 
-\frac{i}{2}g_{\mu\nu}\psi^{\nu}
\eeq
as canonically conjugated variables to $x^{\mu}$ and $\psi^{\mu}$. However, we have constraints
$$\tilde\Phi_{a}=\tilde\pi_{\mu}+\frac{i}{2}g_{\mu\nu}\psi^{\nu}=0,$$
and one needs to use the  Dirac formalism of Poisson brackets with constraints (see, e.g., \cite{Rumpf:1981xh,Braden:1986zu}). Due to the simple nature of the flat connection $\nabla^{-}$, it is convenient to use variables $\psi^{a}$ instead of $\psi^{\mu}$.  Namely, using
\eqref{C-symbol}, we have
\beq\label{L-susy-0}
\LL=\frac{1}{2}g_{\mu\nu}\dot{x}^{\mu}\dot{x}^{\nu}+\frac{i}{2}g_{ab}\psi^{a}\dot\psi^{b},
\eeq
so defining
$$p_{\mu}=\frac{\pa \LL}{\pa \dot{x}^{\mu}}=g_{\mu\nu}\dot{x}^{\nu}\quad\text{and}\quad\pi_{a}=\frac{\pa \LL}{\pa\dot{\psi}^{a}}=-\frac{i}{2}g_{ab}\psi^{b}$$
we see that the fermion degrees of freedom decouple. The corresponding constraints are 
$$\Phi_{a}=\pi_{a}+\frac{i}{2}g_{ab}\psi^{b}=0,$$
and the matrix 
$$C_{ab}=\{\Phi_{a},\Phi_{b}\}=ig_{ab}$$
does not depends on the coordinates $x^{\mu}$. Thus in the Dirac formalism  the canonical coordinates of the reduced phase space are $p_{\mu}$, $x^{\mu}$ and $\psi^{a}$ 
with the following non-vanishing Poisson
brackets 
\beq\label{PB-0}
\{p_{\mu}, x^{\nu}\}=\delta^{\nu}_{\mu}\quad\text{and}\quad\{\psi^{a},\psi^{b}\}=ig^{ab}.
\eeq

For the corresponding Hamiltonian, we have from \eqref{L-susy-0}
\beq\label{susy-H-0}
H=\dot{x}^{\mu}p_{\mu}+\dot\psi^{a}\pi_{a}-\LL=\frac{1}{2}g^{\mu\nu}p_{\mu}p_{\nu},
\eeq
and as in Section \ref{free on G} we obtain
\beq\label{susy-H}
H=\frac{1}{2}g^{ab}l_{a}l_{b}.
\eeq
To establish the classical supersymmetry algebra, it is convenient to write
\beq \label{eq:laginv}
\mathcal L=\frac{1}{2}( J,J)+\frac{i}{2} (\psi,\dot\psi),
\eeq

where $(u,v)=g_{ab}u^{a}v^{b}$ and to represent $\mathrm{ad}$-invariance of the Killing form as
\beq\label{K-inv}
( [x,y],z)=( x,[y,z]),
\eeq
where $x,y,z\in\frak{g}$ can be even or odd elements. Also, when $x,y$ are odd elements of $\frak{g}$, we have 
$$[x,y]=[y,x]\quad \text{and}\quad[x,[x,x]]=0,$$
so
\beq\label{four-psi}
( [x,x],[x,x])=(x,[x,[x,x]])=0.
\eeq

Using equations \eqref{psi-susy}--\eqref{J-susy} and \eqref{K-inv}, we readily compute
\begin{align*}
\delta \LL &=\frac{1}{2}(\delta J,J)+\frac{1}{2}( J,\delta J) +\frac{i}{2}(\delta\psi,\dot\psi)-\frac{i}{2}(\psi,\delta\dot\psi)\\
&=\frac{1}{2}\frac{d}{dt}\left(i(\psi,J)-\frac{1}{6}(\psi,[\psi,\psi])\right).
\end{align*}

so by the Noether theorem, the supercharge $Q$ is given by
\begin{align*}
iQ &= \delta x^{\mu}\frac{\pa \LL}{\pa\dot{x}^{\mu}}+\delta\psi^{a}\frac{\pa \LL}{\pa\dot\psi^{a}}-\frac{1}{2}\left(i(\psi,J)-\frac{1}{6}(\psi,[\psi,\psi])\right)\\
&=i(\psi,J)-\frac{1}{6}(\psi,[\psi,\psi]).
\end{align*}

Using the Legendre transform, we get the following formula for the supercharge in the phase space
\beq\label{Q-phase}
Q=(\psi,L)+\frac{i}{6}(\psi,[\psi,\psi]),
\eeq
where, according to \eqref{e-functions}, 
$$L=l^{a}T_{a},\quad l^{a}=g^{ab}l_{b}.$$
It follows from \eqref{PB-II}  that
$$\{Q,r_{a}\}=0.$$
The following result manifests classical supersymmetry.
\begin{lemma} \label{Q-H-cl} The following classical supersymmetry algebra relation holds
\beq\label{Q-H}
\{Q,Q\}=2iH.
\eeq
\end{lemma}

\begin{proof}
Indeed, we have, using \eqref{PB-0}, \eqref{susy-H} and \eqref{four-psi},
\begin{align*}
\{Q,Q\} & =\{\psi^{a}l_{a},\psi^{b}l_{b}\} 
+\frac{i}{6}\left(l_{a}\{\psi^{a},(\psi,[\psi,\psi])\} + l_{a}\{(\psi,[\psi,\psi]),\psi^{a}\}\right)\\
&\quad-\frac{1}{36}\{(\psi,[\psi,\psi]),(\psi,[\psi,\psi])\}\\
&=ig^{ab}l_{a}l_{b}+( [\psi,\psi],L) -( [\psi,\psi],L) -\frac{1}{4}( [\psi,\psi],[\psi,\psi])\\
&=2i H.\qedhere
\end{align*}
\end{proof}
Since bosonic and fermionic degrees of freedom in the Lagrangian \eqref{L-susy-0} are totally decoupled, quantization of the classical system is straightforward. 

First, the Hilbert space is
\beq
\cH= L^2(G)\otimes \cH_{F},
\eeq
where, according to the Remark \ref{even-odd}, $\cH_{F}$ is irreducible Clifford algebra module of dimension $2^{n/2}$.

Second, classical variables $l_{a}$ and $\psi^{a}$ are replaced by quantum operators $\hat{l}^{a}$ and $\hat\psi^{a}$ acting on $\cH$ and satisfying the following graded commutation relations
\beq \label{eq:gcan}
[\hat \psi^a,\hat \psi^b]= g^{ab}\quad\text{and}\quad
[\hat l_a, \hat l_a]=-i f^c_{ab}\hat l_c
\eeq
--- quantization of the graded Poisson brackets in Appendices \ref{appsec:symp} and \ref{free on G}. In particular, $\hat{l}_{a}=-ie_{a}$, where $e_{a}$ are left-invariant vector fields on $G$, realized as the first order differential operators on $G$. Correspondingly, 
\beq\label{laplace}
\Delta=g^{ab}\hat{l}_{a}\hat{l}_{b}=-g^{ab}e_{a}e_{b}
\eeq
is the Laplace operator on $G$, the quadratic Casimir operator $C_{2}$ on the universal enveloping algebra $U\frak{g}$, realized as algebra of differential operators on $L^{2}(G)$.

The quantum supercharge $\hat{Q}$ is defined unambiguously (i.e., there is no problem of ordering the quantum operators) by the following formula
\beq \label{eq:qg}
\hat Q=\hat \psi^a \hat l_a+{i\ov 6}f_{abc}\hat \psi^a\hat \psi^b\hat \psi^c,\quad f_{abc}=g_{ae}f^{e}_{bc}.
\eeq
Since $\hat{l}_{a}$ and $\hat{r}_{a}=-if_{a}$ commute, where $f_{a}$ are right-invariant vector fields on $G$, realized as the first order differential operators on $G$, we have
$$[\hat{Q},\hat{r}_{a}]=0.$$

The following result is fundamental.
\begin{proposition} Defined by the supersymmetry algebra the quantum Hamiltonian $\hat{H}=\hat{Q}^{2}$ is given by
\beq \label{eq:hs}
\hat H
=\frac{1}{2}g^{ab}\hat l_{a}\hat l_{b}+{1\ov 48}f_{abc}f^{abc}\hat I={1\ov 2}\Delta+{R\ov 12}\hat I,
\eeq
where $\Delta$ is the Laplace operator on $L^2(G)$ 
and
$$R={1\ov 4}f_{abc}f^{abc}=\frac{n}{4}$$ 
is the scalar curvature of the Cartan-Killing metric on $G$.
\end{proposition}
\begin{proof}
From the commutation relations in \eqref{eq:gcan} we readily obtain
\begin{align*}
 \,[\hat \psi^a \hat l_a,\hat \psi^b\hat l_b]&=g^{ab} \hat l_a \hat l_b-if_{ab}^c \hat l_c\hat \psi^a \hat \psi^b,
\\ [\hat \psi^a \hat l_a,f_{bcd}\hat \psi^b\hat \psi^c\hat \psi^d]&=3f^{a}_{bc}\hat l_a \hat \psi^b \hat \psi^c,
\\ [f_{abc}\hat \psi^a\hat \psi^b\hat \psi^c,f_{def}\hat \psi^d\hat \psi^e\hat \psi^f]&=-{3\ov 2}f_{abc}f^{abc}.
\end{align*}
Using these relations, it is straightforward to verify that
\begin{align*}
\hat{Q}^{2} &={1\ov 2}\,[\hat Q,\hat Q ]={1\ov 2}[\hat \psi^a \hat l_a+{i\ov 6}f_{abc}\hat \psi^a\hat \psi^b\hat \psi^c,\hat \psi^d\hat l_d+{i\ov 6}f_{def}\hat \psi^d\hat \psi^e\hat \psi^f]\\&={1\ov 2}g^{ab}\hat l_a \hat l_b+{1\ov 48}f_{abc}f^{abc}\hat I.\qedhere
\end{align*}
\end{proof}
\begin{remark} It should be noted that in representation theory the quantum supercharge $\hat{Q}$ is just the Kostant's cubic Dirac operator on $G$, introduced in \cite{kostant1999cubic} (see also \cite{alekseev2000non}). According to \cite{bismut1989local}, it is just a $G$-invariant Dirac operator, associated with the invariant flat connection $\nabla^-$.
In the physics literature, such operator was first obtained by H.W. Braden \cite{Braden:1986zu} for canonical quantization of the $N=1/2$ supersymmetric sigma model. Our elementary derivation of the supersymmetry algebra (cf. with \cite{Braden:1986zu}, \cite{bismut2008hypoelliptic}) is possible due to commutation relations \eqref{eq:gcan}.
\end{remark}
\begin{remark}
It is quite remarkable that proportional to the scalar curvature innocuously looking constant term in $\hat H$  is precisely the well-known `notorious' DeWitt term for the bosonic quantum particle on curved manifold \cite{dewitt1957dynamical}. It is proportional to $\hbar^{2}$ quantum correction to the action in the path integral for the propagator of a pure bosonic quantum particle moving on a curved manifold, and is obtained by a tedious analysis of the discretized path integral. It is really striking that such term naturally appears through the quantum supersymmetry algebra!  This clearly supports the idea of using supersymmetry to derive pure bosonic trace formulas, and will be a crucial ingredient in the derivation of the trace formula on $G$. 
\end{remark}
\subsection{Eskin's Trace Formula on $G$}\label{E on G}
The supersymmetric particle on $G$ with the connection $\nabla^-$ considered satisfies all conditions of our new localization principle formulated in Section \ref{new-susy-loc}. Namely, the system has $n$ fermion zero modes
$$\chi^a={1\ov \be}\int_0^\be  \psi^a(\tau) d\tau$$
with respect to some basis $T_{1},\dots, T_{n}$ of $\frak{g}$, which we can choose to be orthonormal with respect to the Cartan-Killing form.
According to Remark \ref{hat-chi}, the zero modes in the path integral are saturated by the insertion of operators $\hat \chi^a=\hat\psi^a(0)$ to the spectral side of the formula \eqref{Tr=Str}.

However, it is more natural to use the Cartan-Weyl basis (see Appendix \ref{appsec:g.1}) for the description of the fermion number operator $(-1)^F$.
%in terms of $\hat\psi^a$ is most natural in terms of the Cartan-Weyl basis (see Appendix \ref{appsec:g})
Namely, we have %for the zero modes\footnote{For the convenience, we are using subscript to label the fermions in the Cartan-Weyl basis.}
\beq
\psi(0)=\psi^{a}(0)T_{a} =  \sum_{j }\psi_{j}\cdot iH_j +\sum_{\al \in R}  \psi_{\al}E_\al,
\eeq
where $i=\sqrt{-1}$ and $\bar\psi_{j}=\psi_{j}$, $\bar\psi_{\al}=-\psi_{-\al}$. Corresponding fermion operators %saturates to the zero modes components 
satisfy anti-commutation relations
$$ [\hat \psi_i,\hat \psi_{j}]=\del_{ij},\quad [\hat \psi_\al,\hat \psi_{\be}]=-\del_{\al,-\be},$$
where $\hat\psi^{\dagger}_{i}=\hat\psi_{i}$ and $\hat\psi^{\dagger}_{\al}=-\hat\psi_{-\al}$. Introduce Hermitian fermion operators $\hat\chi_{j}=\hat\psi_{j}$ and $\hat\chi_{\al}$, $\hat\chi_{-\al}$ by the following formulas
$$ \hat\chi_{\al}=\frac{1}{\sqrt{2}}(\hat\psi_{\al}-\hat\psi_{-\al}), \quad  \hat\chi_{-\al}=\frac{1}{i\sqrt{2}}(\hat\psi_{\al}+\hat\psi_{-\al}),\quad\text{where}\quad \al\in R_{+}.$$
 The operators $\{ \hat \chi_j , \hat\chi_{\al} , \hat\chi_{-\al}\}$ satisfy canonical relations \eqref{cac}. For each two-dimensional fermion Hilbert subspace generated by $\hat\psi_\al$ and $\hat \psi_{-\al}$, we have a natural fermion number operator 
$$(-1)^{F_\al}
%=\left(\hat\chi_{\al}\hat\chi_{-\al}-\frac{1}{2}\hat{I}\right)
=2i\hat\chi_{\al} \hat\chi_{-\al},\quad\al\in R_{+}.$$
%= \delta_{ij}
%[\hat \psi_i,\hat \psi_j ]=\del_{ij},\quad [\hat \psi_\al,\hat \psi_{\be}]=-\del_{\al,-\be}.
Thus the fermion number naturally associated with the Cartan-Weyl basis is
\beq \label{eq:GF}
c_n\hat \psi^1\dots \hat \psi^n=  c_r \hat \chi_1 \dots \hat \chi_r   \prod_{\al \in R_+}i \hat\chi_{\al} \hat\chi_{-\al}=2^{-n/2}(-1)^{F},
\eeq
where the phase $c_r$ is such that \eqref{F-square} holds. So according to \eqref{eq:hs},  bosonic and fermionic degrees of freedom are totally decoupled and
$$\text{Str}\Big(c_r \hat \chi_1 \dots \hat \chi_r   \prod_{\al \in R_+}i \hat\chi_{\al} \hat\chi_{-\al}e^{-\be \hat H}\Big)=e^{-\frac{1}{12}\be R} \Tr e^{-\frac{1}{2}\be \Delta}.$$

Since $\Delta$ commutes with left and right translations, one can also express the heat kernel on $G$ as a supertrace. Namely, we recall that the heat kernel is a fundamental solution  $K_{\tau}(g_{1},g_{2})$  of the heat equation on $G$
$$\frac{\pa K}{\pa\tau}=-\frac{1}{2}\Delta K$$
with respect to $g_{1}$, satisfying
$$\lim_{\tau\to 0}K_{\tau}(g_{1},g_{2})=\delta_{G}(g_{1}g^{-1}_{2}),$$
where $\del_{G}$ is the Dirac delta-function on $G$ with respect to the Cartan-Killing volume form. Fix a Cartan subgroup $T$ in $G$ and corresponding Cartan subalgebra $\frak{t}$ in $\frak{g}$, $\dim\frak{t}=r$, the rank of $\frak{g}$. It follows from the bi-invariance of the heat kernel that it only depends on $g_{1}g^{-1}_{2}\in T$, and we will denote it by $k_{t}(e^{h})$, where $h\in\frak{t}$ and $e^{h}\in T$. Using Dirac notation,
$$k_{\tau}(e^{h})=K_{\tau}(e^{h},1)=\la e^{h}|e^{-\frac{1}{2}\tau\Delta}|1\ra,$$
where $1$ is the identity element in $G$.
 
 Correspondingly, $ \Tr e^{-\frac{1}{2}\be \Delta}=V_{G}k_{\beta}(1)$, where $V_{G}$ is the volume of $G$, and more generally
 $$ \Tr e^{-\frac{1}{2}\be \Delta +i( h,\hat{r})}=V_{G}k_{\beta}(e^{h}),$$
 where $\hat{r}=\hat{r}^{a}T_{a}$. The extra term $i( h,\hat{r})$ in the exponent can be thought as a (imaginary) `chemical potential'  added to the Hamiltonian $\hat{H}$. Since the operators $\hat{r}^{a}$ commute with $\hat{Q}$, we have
 \beq \label{eq:eskinstr}
\text{Str}\Big( c_r \hat \chi_1 \dots \hat \chi_r   \prod_{\al \in R_+}i \hat\chi_{\al} \hat\chi_{-\al}e^{-\be \hat H +i( h,\hat{r}) }\Big)=V_G e^{-\frac{1}{12}\be R}k_{\be}(e^{h}). 
\eeq
Here the supertrace is given by the following path integral
\beq \label{eq:strpi}
\int \limits_{\Pi TL G}  \chi_1 \dots \chi_r  \prod_{\al \in R_+} i \chi_{\al} \chi_{-\al} \,e^{-S_E^{h}}\, \cD g \cD\psi\, ,
\eeq
with the Euclidean action
\beq 
S_E^{h}= {1\ov 2} \int_0^\be \left((J,J)+(\psi, \dot \psi)\right)d\tau + {1\ov \be}\int_{0}^{\be}(\mathrm{Ad}_{g^{-1}}h,J)d\tau+{1\ov 2\be}(h,h)
\eeq
and $G$-invariant `measure' $\cD g$. 
The fermion `measure' $\cD\psi$ is determined from the condition
\beq\label{fermi-measure}
\begin{gathered}
1=\Tr_{\mathscr H_F}\Big(c_r \hat \chi_1 \dots \hat \chi_r   \prod_{\al \in R_+}i \hat\chi_{\al} \hat\chi_{-\al}\,(-1)^{F}\Big)\\=\int \limits_{\Pi TL G}  \chi_1 \dots \chi_r  \prod_{\al \in R_+} i \chi_{\al} \chi_{-\al} \,e^{-\frac{1}{2}\int_{0}^{\be}(\psi,\dot\psi)d\tau}\, \cD\psi, 
\end{gathered}
\eeq
which follows from  \eqref{eq:GF}.
%The measure for the fermion is determined from the definition \eqref{eq:GF} and the footnote \ref{ft:measure} as follows. 

Indeed, using the Cartan-Weyl basis we have,
\beq\label{psi}
\psi(\tau)=  \sum_{n=-\infty}^{\infty}\sum_{j=1}^{r}\psi_{j,n}\cdot i H_j e^{i\omega_n \tau}+\sum_{n=-\infty}^{\infty}\sum_{\al\in R}\psi_{\al,n}E_\al e^{i\omega_n\tau},
\eeq
where $\overline{\psi}_{j,n}=\psi_{j,-n}$, $\overline{\psi}_{\al,n}=-\psi_{-\al,-n}$, so
$$-\frac{1}{2}\int_{0}^{\be}(\psi,\dot\psi)d\tau=\sum_{n=1}^{\infty}\sum_{j=1}^{r}i\be\omega_{n}\psi_{j,n}\overline{\psi}_{j,n}+\sum_{n=-\infty}^{\infty}\sum_{\al\in R_+}i\be\omega_{n}\psi_{\al,n} \overline{\psi}_{\al,n}$$
Now if the fermion measure is chosen to be
%The fermionic measure consistent with our definition of $(-1)^F$ is given by
\beq \label{eq:fmG}
\mathscr D\psi=(-1)^{r(r-1)/2}\left( \prod_{j=1}^r d\psi_{j,0} \right)\left(\prod_{j=1}^r \prod_{n=1}^\infty id\psi_{j,n} d\overline{\psi}_{j,n}  \right) \left(\prod_{\al\in R_+} \prod_{n\in \mathbb Z}  id\psi_{\al,n} d\overline{\psi}_{\al,n}\right),
\eeq
then it follows from \eqref{psi} then the only terms in $\chi_1 \dots \chi_r  \prod_{\al \in R_+} i \chi_{\al} \chi_{-\al}$ that give non-zero contribution to the path integral come from $\psi_{j,0}$ 
and $\psi_{\al,0}$ so using the rules of the fermion integration and the formula
$$\prod_{n=1}^{\infty}(2\pi n)=1,\quad i \prod_{n=1}^{\infty}(-2\pi n)=1,$$
where the 2nd identity comes from
\beq \label{eq:-1}
\prod_{n=1}^\infty z =e^{\zeta(0)\log z}= 1/\sqrt{z},\quad -\pi <\arg z\leq \pi,
\eeq
we get \eqref{fermi-measure}.

%which ensures $\Tr_{\mathscr H_F}[c_n\hat \chi^1\dots \hat \chi^{n}]=1$ from the path integral side. 

Now let us return to the full path integral \eqref{eq:strpi}.
Because of the additional term $(\mathrm{Ad}_{g^{-1}}h,J)$ in the action, it is easy to verify that $S_E^{h}$ is invariant under the
 modified supersymmetry transformation $\del_{h}$, which in the Euclidean time for fixed $h\in\frak{t}$ has the form
\begin{align*}
\del_{h} g &=g\psi,\\
\del_{h}\psi &=-J^{h}-\psi\psi,\\
\del_{h} J^h &=(\pa_\tau+\text{ad}_{J^{h}})\psi,
\end{align*}
where $J^{h}=J+\dfrac{1}{\be}\text{Ad}_{g^{-1}}h$. 

The spectral representation of $k_{\be}(e^{h})$ is well-known since $\Delta=C_{2}$, the quadratic Casimir operator. Namely, we have
\beq\label{spectral-trace-1}
k_{\be}(e^{h})=\frac{1}{V_{G}}\sum_{\lambda\in\mathrm{Irrep}\, G}d_{\lambda}\,\chi_{\lambda}(h)e^{-\frac{1}{2}\be C_{2}(\lambda)},
\eeq
where $\chi_{\lambda}$ is the character for an irreducible representation $\lambda$ of $G$,
$d_{\lambda}$ is its dimension, and $C_{2}(\lambda)$ is the eigenvalue of $C_{2}$. We have (see Appendix \ref{appsec:g.1} for notation)

$$C_{2}(\lambda)=\la\lambda +\rho,\lambda+\rho\ra -\la\rho,\rho\ra,\quad\text{where}\quad \rho=\frac{1}{2}\sum_{\al\in R_{+}}\al\in P_{+},$$
where $R_{+}$ is the set of positive roots and $P_{+}$ is the set of dominant weights.
Using Weyl dimension formula
$$d_{\lambda}=\prod_{\al\in R_{+}}\frac{\la\lambda+\rho,\al\ra}{\la\rho,\al\ra},$$
one can rewrite \eqref{spectral-trace-1} as
\beq\label{spectral-trace-2}
k_{\be}(e^{h})=\frac{1}{V_{G}}e^{\frac{1}{2}\beta\la\rho,\rho\ra}\sum_{\lambda\in P_{+}}\chi_{\lambda}(h)\prod_{\al\in R_{+}}\frac{\la\lambda+\rho,\al\ra}{\la\rho,\al\ra}e^{-\frac{1}{2}\be (\lambda+\rho,\lambda+\rho)}.
\eeq
The Freudenthal-de Vries's `strange formula' 
$$\ex{\rho,\rho}=\dfrac{n}{24},$$ 
ensures that the prefactor in \eqref{spectral-trace-2} exactly cancels the factor $e^{-\frac{1}{2}\be R}=e^{-\frac{1}{48}n\be}$ in  \eqref{eq:eskinstr}!

It is remarkable, that there is another representation for $k_{\be}(e^{h})$ as a sum over so-called characteristic lattice, obtained by L.D. Eskin \cite{MR0206535}. Namely, suppose (which is the main assumption)  that Cartan element $h$ is regular, which means that $\la h,\al\ra\notin 2\pi i\ZZ$ for any root $\al\in R$. Then one has the following representation
\beq\label{co-root-formula}
k_{\be}(e^{h})=\frac{e^{\frac{1}{2}\beta\la\rho,\rho\ra}}{(2\pi \be)^{n/2}}\sum_{\ga \in\Gamma}\prod_{\al\in R_+}{\frac{1}{2}\ex{\al,h+\gamma}\ov \sinh\frac{1}{2}\ex{\al,h+\gamma }}e^{-{ \frac{1}{2\be}(h+\gamma,h+ \gamma)}},
\eeq
where $\Gamma=\{\gamma\in \mathfrak t :e^\gamma=1\}$ is the characteristic lattice\footnote{In case $G$ is simple and simply connected, $\Gamma=2\pi iQ^{\vee}$ where $Q^\vee$ is a coroot lattice.}, which is related to the maximal torus of $G$ by $T=\frak{t}/\Gamma$. Comparing formulas \eqref{spectral-trace-2} and \eqref{co-root-formula} yields Eskin trace formula\footnote{Note that Eskin was using the convention that $R\subset\frak{t}$, which explains the difference between
hyperbolic sine functions in \eqref{co-root-formula} and trigonometric sine functions in \cite{MR0206535}.}.

Up to a DeWitt term $e^{\frac{1}{2}\beta\la\rho,\rho\ra}$, the right hand side in  \eqref{co-root-formula} exhibits a nature of the exact semi-classical approximation with geodesics $\exp((h+\ga) \tau/\be)$, $\ga\in\Gamma$, connecting $1$ and $e^{h}$. 
The regularity condition ensures that each denominator in \eqref{co-root-formula} is non zero, since for $\ga\in\Gamma$ 
\begin{equation}\label{prod-al-ga}
\prod_{\al \in R_+}\sinh \tfrac{1}{2}\la\al,h+\ga\ra = \pm\prod_{\al \in R_+}\sinh \tfrac{1}{2}\la\al,h\ra.
\end{equation}  

Here we show that the Eskin trace formula can be obtained very naturally as a result of supersymmetric localization of the path integral \eqref{eq:strpi}. Namely, we have the following result.

\begin{theorem} \label{E-TF} We have the following exact localization formula
 \begin{gather*}
 \int \limits_{\Pi TL G} \chi_1 \dots \chi_r   \prod_{\al \in R_+} {\chi}_{-\al} \chi_{\al}e^{-S_E^{h}}\, \cD g \cD\psi \\
 =\frac{V_{G}}{(2\pi \be)^{n/2}}\sum_{\ga \in\Gamma}\prod_{\al\in R_+}{\frac{1}{2}\ex{\al,h+\gamma}\ov \sinh\frac{1}{2}\ex{\al,h+\gamma }}e^{-{ \frac{1}{2\be}(h+\gamma,h+ \gamma)}}.
 \end{gather*}
\end{theorem}
\begin{proof}

Put
\begin{align}
V & =-\frac{1}{2}\int_0^\be(\dot J^{h},\dot \psi)\,d\tau,\label{V-1}\\
\intertext{so}
\del_{h} V&={1\ov2}\int_0^\be \left( (\dot J^{h},\dot J^{h})+(\dot \psi,(\pa_\tau+\text{ad}_{J^{h}})\dot \psi)\right)\,d\tau.\label{del-V}
\end{align}
It is straightforward to check that $\del_{h}^2 V=0$ and $V, \del_{h}V\perp \chi^a $, $a=1,\dots, n$. Hence $V$ is an invariant deformation to the path integral \eqref{eq:strpi}, and  according to Proposition \ref{new-loc},
\beq \label{eq:gploc}
\begin{gathered}
\int \limits_{\Pi TL G} \chi_1 \dots \chi_r   \prod_{\al \in R_+}  i \chi_{\al} \chi_{-\al}  e^{-S_E^{h}}\cD g\cD\psi \\=\int \limits_{\Pi TL G} \chi_1 \dots \chi_r   \prod_{\al \in R_+}i \chi_{\al} \chi_{-\al}  e^{-S_E^{h}-s\delta_{h}V} \cD g\cD\psi.
\end{gathered}
\eeq
Moreover, purely bosonic part of $\delta_{h}V$ is positive semi-definite with zeros at $\dot J^h=0$, so in the limit $s\rightarrow \infty$, the path integral \eqref{eq:gploc} localizes on the locus $\dot J^h=0$. 
To facilitate the computation, we use the left $G$-invariance and integrate over fermionic zero modes with the measure \eqref{eq:fmG} to obtain
\beq \label{eq:spec}
\begin{gathered}
\int \limits_{\Pi TL G} \chi_1 \dots \chi_r  \prod_{\al \in R_+}  i \chi_{\al} \chi_{-\al}  e^{-S_E^{h}-s\delta_{h}V} \cD g\cD\psi \\=V_{G} \bigg(\prod_{\al\in R_+} i\bigg)\int \limits_{\Pi T\Omega G}e^{-S_E^{h}-s\delta_{h}V} \cD' g\cD'\psi,
\end{gathered}
\eeq
where $\Omega G$ is the space of based loops on $G$, $g(0)=g(\be)=1$. 
Since $h\in\frak{t}$ is regular,  solutions of the equation $\dot J^h=0$ in $\Omega G$ are isolated geodesics parameterized by the lattice $\Gamma$,
\beq
g_\ga(\tau)=e^{\frac{1}{\be}\tau\ga}\quad\text{and}\quad J^{h}_{\ga}=\frac{h+\ga}{\be},\quad \ga\in\Gamma.
\eeq

Indeed, we have
$$J^h =g^{-1}\left(\dot{g}+\frac{1}{\be}hg\right)=\tilde g^{-1} \dot{\tilde g},$$ 
where $\tilde g(\tau)=e^{h\tau/\be}g(\tau)$. Hence equation $\dot J^h=0$ implies $J^{h}(\tau)=c/\be$ with some $c\in\frak{g}$, so $\tilde{g}(\tau)=\tilde{g}(0)e^{c\tau/\be}$ and $g(\tau)=e^{-h\tau/\be}g(0)e^{c\tau/\be}$. The condition $g\in \Omega G$ gives
$e^{h}=e^{c}$, and $h\in\frak{t}$ being regular implies that $c\in\frak{t}$ (see Lemma \ref{lemma:exp} in Appendix \ref{appsec:g.1}). Thus $e^{c-h}=1$ and $c-h=g\tilde\gamma g^{-1}$ for some $\tilde\gamma \in \Gamma$ and $g\in G$, and since $c-h\in\frak{t}$, we have $g\in W$, the Weyl group of $G$,
and $\gamma=g\tilde\gamma g^{-1}\in\Gamma$.

Similarly, since $h$ is nonsingular, it is easly to see that solution of the equation
$$\pa_{\tau}(\pa_\tau+\text{ad}_{J^{h}})\dot \psi=0$$
in $\Pi\Omega\frak{g}$ gives $\psi_{cl}=0$.

Now as in Section \ref{sec:S1}, using \eqref{del-V} we readily obtain
\beq  \label{eq:inter}
\begin{gathered}
\lim_{s\to\infty}\int \limits_{\Pi T\Omega G}e^{-S_E^{h}-s\delta_{h}V} \cD'g\cD'\psi \\=\frac{1}{(2\pi)^{\frac{n}{2}}}
\int \limits_{\Pi T\Omega G} e^{-S_{E}^{h}[g,\psi]}\del(\dot J^{h})\del (\psi)\Pf{-\pa^{3}_\tau-\text{ad}_{J^{h}}\pa^{2}_{\tau}} \cD' g\cD'\psi  \\
=\frac{1}{(2\pi)^{\frac{n}{2}}}\int \limits_{\Omega G} \sum_{\ga\in\Gamma}e^{-S_{E}^{h}[g,0]}\frac{\del(g g^{-1}_{\ga})}{|\det (\text{d}\dot{J}^{h})|}\Pf{-\pa^{3}_\tau-\text{ad}_{J^{h}}\pa^{2}_{\tau}}\cD' g\\
 = \frac{1}{(2\pi)^{\frac{n}{2}}}\sum_{\ga\in\Gamma}\frac{\Pf{-\pa^{3}_\tau-\text{ad}_{(h+\ga)/\be}\pa^{2}_{\tau}}}{|\det (\text{d}\dot{J}^{h})|} e^{-{(h+\ga,h+\ga)\ov 2\be}},
\end{gathered}
\eeq
where the linear operator $\text{d}\dot{J}^{h}:\Omega\frak{g}\to \Omega\frak{g}  $ is the differential of the functional $\dot{J}^{h}(\tau)$ on $\Omega G$ at the critical point $J^{h}_{\ga}$.

By definition of $J^{h}$
we have
\beq\label{dJ}
\text{d}J^{h}(\tau)=\pa_\tau+\text{ad}_{J^h(\tau)},
\eeq
so
\beq\label{d-dotJ}
\text{d}\dot{J}^{h}(\tau)=\pa^{2}_\tau+\pa_{\tau}\text{ad}_{J^h(\tau)}.
\eeq
Indeed, put $\del g=gX$, $X\in\frak{g}$, so
$$\del\dot{g}=\dot{g}X+g\dot{X}$$
and for $\del J^{h}=\text{d}\dot{J}^{h}(\tau)(X)$ we have
$$\del J^{h}=-g^{-1}\del g g^{-1}\dot{g}+g^{-1}\del\dot{g}+\frac{1}{\be}(g^{-1}h\delta g-g^{-1}\delta g g^{-1}hg)=\dot{X}+[J^{h}, X],$$
which proves \eqref{dJ}.
Thus for each critical point $J^{h}(\tau)=(h+\ga)/\be$ we obtain
\beq
 \text{d}\dot J^{h}(\tau)=\pa^{2}_\tau+\text{ad}_{(h+\ga)/\be}\pa_{\tau}.
\eeq

To compute the functional determinants in \eqref{eq:inter}, we use the Cartan-Weyl basis \eqref{g-real} for $\frak{g}$ (cf. \cite{Picken:1988ev}) and 
as in Section \ref{sec:S1} consider the expansions
\begin{align}
X(\tau) & =\sum_{n\neq 0}\left(\sum_{j=1}^{r}c_{j,n}iH_{j}+\sum_{\al\in R_+}c^{(1)}_{\al,n}X_{\al}+\sum_{\al\in R_+}c^{(2)}_{\al,n}Y_{\al}\right)u_{n}(\tau)\in L\frak{g}
/\frak{g}\simeq \Omega\frak{g}, \label{expansion-real-1}\\
\psi(\tau) &=\sum_{n\neq 0}\left(\sum_{j=1}^{r}\psi_{j,n}iH_{j}+\sum_{\al\in R_+}\psi^{(1)}_{\al,n}X_{\al}+\sum_{\al\in R_+}\psi^{(2)}_{\al,n}Y_{\al}\right)u_{n}(\tau)\in\Pi\Omega\frak{g}, \label{expansion-real-2}
\end{align}
where we put $X_{\al}=E_{\al}-E_{-\al}$ and $Y_{\al}=i(E_{\al}+E_{-\al})$.  
It follows from the commutation relations
$$[H,X_{\al}]=-i\la H,\al\ra Y_{\al}\quad\text{and}\quad[H,Y_{\al}]=i\la H,\al\ra X_{\al},\quad \text{where}\quad H\in i\frak{t},$$
(see Appendix \ref{appsec:g.1})  that $iH_{j}u_{n}(\tau)$ and $X_{\al}u_{n}(\tau)+Y_{\al}u_{-n}(\tau)$ are the eigenfunctions of the operator  $-\pa^{2}_\tau-\text{ad}_{(h+\ga)/\be}\pa_{\tau}$  with the eigenvalues
$\omega^{2}_{n}$ and $\omega^{2}_{n} +i\dfrac{\ex{h+\ga,\al}}{\be}\omega_{n}$.

Similarly, in the two-dimensional subspaces  spanned by $iH_{j}u_{n}(\tau), iH_{j}u_{-n}(\tau)$ and by $X_{\al}u_{n}(\tau)+Y_{\al}u_{-n}(\tau), X_{\al}u_{-n}(\tau)-Y_{\al}u_{n}(\tau)$, where $n\in \mathbb Z$ and $\al\in R_+$, 
the skew-symmetric operator $-\pa^{3}_\tau-\text{ad}_{(h+\ga)/\be}\pa^{2}_{\tau}$ acts by the following $2\times 2$ matrices
$$\begin{pmatrix} 0 & \omega_{n}^{3}\\-\omega_{n}^{3} & 0\end{pmatrix}\quad \text{and}\quad  \begin{pmatrix} 0 & \omega_{n}^{3}+i\omega^{2}_{n}\dfrac{\ex{h+\ga,\al}}{\be}\\-\omega_{n}^{3}-i\omega^{2}_{n}\dfrac{\ex{h+\ga,\al}}{\be} & 0\end{pmatrix}.$$
Finally, for the ratio of the functional determinants in \eqref{eq:inter} we have, using the standard zeta function regularization,
\begin{gather*}
\frac{\Pf{-\pa^{3}_\tau-\text{ad}_{(h+\ga)/\be}\pa^{2}_{\tau}}}{\det(-\pa^{2}_\tau-\text{ad}_{(h+\ga)/\be}\pa_{\tau})} =\\ \left(\prod_{n=1}^{\infty} \omega^{3}_n\!\right)^{\! r}\!\!\!\prod_{\al\in R_+}\prod_{n\neq 0}^{\infty} \left(\!\omega^{3}_n+i\frac{\ex{h+\ga,\al}}{\be}\omega^{2}_{n}\right) \Big/
\left(\prod_{n=1}^{\infty} \omega^{4}_n\!\right)^{\! r}\!\!\!\prod_{\al\in R_+}\prod_{n\neq 0}^{\infty}\left(\!\omega^{2}_n+i\frac{\ex{h+\ga,\al}}{\be}\omega_{n}\!\right)^{\! 2}\\
=\left(\prod_{n=1}^{\infty} \omega_n\right)^{\! -r}\prod_{\al\in R_+}\,\prod_{n\neq 0} \left(\omega_n+i\frac{\ex{h+\ga,\al}}{\be}\right)^{\!-1}\\
=\be^{-\frac{n}{2}}\prod_{\al \in R_+}\frac{\frac{1}{2}\la\al,h+\ga\ra}{ i \sinh\frac{1}{2}\la\al,h+\ga\ra}, %\qedhere
\end{gather*}
where in the last line we used \eqref{eq:-1}.

Finally, together with the factor of $i$'s in r.h.s. of \eqref{eq:spec} from the zero modes integrals, we provide the path integral derivation Eskin trace formula which is Theorem \ref{E-TF}.
\end{proof}
\begin{remark} Using \eqref{eq:eskinstr}, \eqref{eq:strpi} and \eqref{spectral-trace-2}, the statement of the theorem can be rewritten as the equality ``Spectral Trace = Matrix Trace'' for the operator
$e^{-\frac{1}{2}\be\Delta}$,
\begin{gather*} 
e^{\frac{1}{2}\beta\la\rho,\rho\ra}\sum_{\lambda\in P_{+}}\chi_{\lambda}(h)\prod_{\al\in R_{+}}\frac{\la\lambda+\rho,\al\ra}{\la\rho,\al\ra}e^{-\frac{1}{2}\be\la\lambda+\rho,\lambda+\rho\ra}\\=V_{G}\frac{e^{\frac{1}{2}\beta\la\rho,\rho\ra}}{(2\pi \be)^{n/2}}\sum_{\ga \in\Gamma}\prod_{\al\in R_+}{\frac{1}{2}\ex{\al,h+\gamma}\ov \sinh\frac{1}{2}\ex{\al,h+\gamma }}e^{-{ \frac{1}{2\be}(h+\gamma,h+ \gamma)}},
\end{gather*}
which generalizes the Jacobi inversion formula.
After trivial rewriting, this formula is precisely the Eskin trace formula! It should also be noted that when $G$ is simply connected the sign in formula 
\eqref{prod-al-ga} is plus for all $\ga\in\Gamma$,
and Eskin formula is proved in Theorem 4.3 in \cite{bismut2008hypoelliptic}. The latter is obtained as a special case of Theorem 4.3.4 in \cite{Frenkel1984}, when one of Cartan elements goes to $0$.
\end{remark}

\section{Singular Trace Formula on $G$}

Since the spectral side in Eskin formula is a smooth function of $h\in\frak{t}$, it is natural to relax the regularity condition on $h$. In particular, for $h=0$ the spectral side is just  $\Tr e^{-\frac{1}{2}\beta\Delta}$. However, for singular $h$ the sum over $\Gamma$ has vanishing denominators, and it was observed in  \cite{MR0206535} that one needs to sum over the orbits of the Weyl group of $G$, and only then take the limit $h\to 0$. Physically, such singular behavior indicates that geodesics are no-longer isolated, so the localization formalism should be modified. In this section, we consider the most singular case $h=0$; generalization to other singular cases is  straightforward. 

Let's start from the localization formula established in the previous section in the case of $h=0$. 
\beq
{1\ov V_G}e^{-\frac{1}{12}\be R} \Tr e^{-\frac{1}{2}\be \Delta}= \int \limits_{\Pi T \Omega G} e^{-S_E -s\del V} \cD g\cD\psi \equiv I(s)
\eeq

We write the path integral as
$$I(s)=\int_{\Pi T\Omega G}F(J,\psi)\exp\left\{- {s \ov 2}\int_0^\be\left((\dot J,\dot J)+ (\dot \psi ,(\pa_\tau+\text{ad}_{J})\dot \psi)\right)d\tau\right\} \cD g\cD\psi  $$
where
$$F(J,\psi)=\exp\left\{- {1\ov 2}\int_0^\be\left((J,J)+(\psi,\dot \psi)\right)d\tau\right\},$$

We first analyze the critical points under the localization limit $s\rightarrow \infty$. 
\subsection{Critical points} We have $\dot{J}=0$, where $J=g^{-1}\dot{g}$, so equation $J=\mathrm{const}$ and $g(0)=g(\be)=1$ give
$g(\tau)=e^{u\tau}$, where $e^{u\beta}=1$. Therefore $u=\frac{1}{\be}v\ga v^{-1}$, where $\ga\in\Gamma$, the characteristic lattice, and $v\in G/C_{\ga}$, where
$$C_{\ga}=\{g\in G: g\ga g^{-1}=\ga\}$$
is the centralizer of $\ga$ in $G$. 
Thus $u\in\mathcal{O}_{\ga}$, the orbit of $\ga/\be$ in $\frak{g}$ under the adjoint action of $G$. Clearly, $\mathcal{O}_{\ga}=\mathcal{O}_{\ga'}$ if $\ga'=w\ga w^{-1}$, where
$w\in W$, the Weyl group of $G$. 
Thus connected components for bosonic critical points are the sets $\mathcal{M}_{\ga}$ in $\Omega G$,
\beq\label{J-critic}
g_{cl}(\tau)=v e^{\frac{\tau}{\be}\ga}v^{-1},\quad J_{cl}=\frac{1}{\be}v\ga v^{-1},
\eeq
parameterized by $v\in G/C_{\ga}$, where $\ga\in\Gamma/W$. In other words, $J_{cl}\in\mathcal{O}_{\ga}$.

As in Section \ref{E on G}, for $g\in \Omega G$ we put $X=g^{-1}\del g$ and rewrite the expansion \eqref{expansion-real-1} in the form\footnote{\label{ft:basis}Note that just for convenience, we are using a slightly different basis $\{H_j, E_\al\}$ compared to the $\{iH_j,E_\al\}$ used in Section \ref{sec:G}.}
\beq\label{X-formula}
X(\tau)= \sum_{n,j}z_{j,n}H_j e^{i\omega_n \tau}+\sum_{n,\al\in R}z_{\al,n}E_\al e^{i\omega_n\tau}\in \Omega\frak{g},
\eeq
where  $\overline{z}_{j,n}=-z_{j,-n}$, $\overline{z}_{\al,n}=-z_{-\al,-n}$ and $H_1,\dots H_{r}$, $E_{\al}$, $\al\in R$, is the orthonormal Cartan-Weyl basis (see Appendix \ref{appsec:g.1}). Condition that
$X$ is a tangent vector to the image of $\Omega G$ under the map $g\to J=g^{-1}\dot{g}$ is $X(0)=0$, so the path integral measure $\cD' g$ becomes proportional to 
\begin{gather*}
\prod_{j}\delta\left(\sum_{n \in\ZZ} z_{j,n}\right)\prod_{\al\in R}\delta\left(\sum_{n\in\ZZ}z_{\al,n}\right) \\
\times \prod_{j} dz_{j,0}\prod_{n=1}^{\infty}dz_{j,n}dz_{j,-n}\prod_{\al\in R}dz_{\al,0}\prod_{n=1}^{\infty}dz_{\al,n}dz_{\al,-n}.
\end{gather*}

For notational convenience, we put $\gamma=2\pi i \nu$.\footnote{When $G$ is simple and simply connected we have  $\nu\in Q^\vee$.}
Writing $v=e^{\sum u^{\al}E_{\al}}$ with $\overline{u}^\al=-u^{-\al} $, where summation goes over $\al$ such that $\ex{\al,\nu}\neq 0$ and using 
\beq\label{ga-al}
[\ga, E_{\al}]=2\pi i \ex{\al,\nu}E_{\al},
\eeq
we have for the corresponding tangent vector along $\mathcal{M}_{\ga}$
\beq \label{eq:tan}
g_{\ga}^{-1}\delta g_{\ga} & = u^\al (e^{-\frac{\tau}{\be}\ga}E_\al e^{\frac{\tau}{\be}\ga}-E_\al)\\
&=u^\al (e^{- i\omega_{\ex{\al,\nu}}\tau}-1)E_\al.
\eeq

The fermion critical points satisfy the equation
\beq\label{fermi-critic}
\pa_{\tau}(\pa_\tau+\text{ad}_{J})\dot \psi=0,
\eeq
where $J=J_{cl}$. We have
$$(\pa_\tau+\text{ad}_{J})\dot \psi=c,$$
and since $\psi\in\Pi\Omega\frak{g}$, using Fourier series we have $c=0$, so
\beq\label{fermi-critic-1}
(\pa_\tau+\text{ad}_{J})\dot \psi=0.
\eeq
It is sufficient to solve this equation for $v=1$. It is convenient to use the expansion for $\psi\in  \Pi\Omega \mathfrak g$ as\footnote{Similarly, we are using a different basis for a convenience, see footnote \ref{ft:basis}.}
\beq\label{psi-omega}
\psi=  \sum_{n,j }\psi_{j,n}H_j e^{i\omega_n \tau}+\sum_{n,\al}\psi_{\al,n}E_\al e^{i\omega_n\tau}\in\Pi\Omega\frak{g},
\eeq
where $\psi_{j,0}=\psi_{\al,0}=0$ and $\overline{\psi}_{j,n}=-\psi_{j,-n}$, $\overline{\psi}_{\al,n}=-\psi_{-\al,-n}$.
Then equation \eqref{fermi-critic-1} gives
$$ \sum_{n,j} i\omega_{n}\psi_{j,n}H_j e^{i\omega_n \tau}+\sum_{n,\al}\left(i\omega_{n}+2\pi i\frac{\langle\al,\nu\rangle}{\be}\right)\psi_{\al,n}E_\al e^{i\omega_n\tau}=0,$$
so that $\psi_{n,j}=0$ and $\psi_{\al,n}=0$ except
$\psi_{\al, -\ex{\al,\nu}}$ for $\ex{\al,\nu}\neq 0$, which are arbitrary. Thus the general solution of  \eqref{fermi-critic-1} is
\beq\label{psi-sol}
\chi=\sum_{\ex{\al,\nu}\neq 0}\psi_{\al,-\ex{\al,\nu}}E_{\al}e^{i\omega_{-\ex{\al,\nu}}\tau}
\eeq
and is parameterized by $\Pi (\frak{g}/c_{\ga})$, where $c_{\ga}=\{x\in\frak{g}: [x,\ga]=0\}$ the centralizer of $\ga$ in $\frak{g}$. 

\subsection{Quadratic form}  It is easy to show that the map $J: \Omega G\to L\frak{g}$ is injective and its image is a coadjoint orbit $\widehat{\mathcal{O}}_{1}$ of the element
 $0\in L\frak{g}$ (see Appendix \ref{A3}).
Correspondingly, a connected component of the critical locus of $J_{cl}$ is a coadjoint orbit $\mathcal{O}_{\ga}$ of $G$ in $\frak{g}$, and we have $\mathcal{O}_{\ga}\subset\widehat{\mathcal{O}}_{1}$. 

According to formula \eqref{dJ}, the differential $\text{d}J$ of the map $J: \Omega G\to L\frak{g}$ is the differential operator $D_{J}=\pa_{\tau}+\text{ad}_{J}$.
Let the vector $X\in \Omega \frak{g}$, given by \eqref{X-formula}, be the tangent vector to $\Omega G$ at $g_{cl}$. Then for $g(\tau)=g_{cl}(\tau)e^{ X(\tau)}$ we obtain
\begin{equation}\label{J-expansion}
J=g^{-1}\dot{g}=J_{cl}+ DX+\frac{1}{2}[DX,X]+\text{higher order terms},
\end{equation}
where we denote $D=D_{J_{cl}}$.

It follows from \eqref{X-formula} that\footnote{Here we are not assuming that constant terms of $X(\tau)$ vanish but rather imposing the constraints $X(0)=X(\tau)=0$, which follows from $g\in\Omega G$.}
$$DX=\sum_{n,j}i\omega_{n}z_{j,n}H_{j}e^{i\omega_{n}\tau}+\sum_{n,\al}i\omega_{n+\ex{\al,\nu}}z_{\al,n}E_{\al}e^{i\omega_{n} \tau}.$$

Using the $\mathrm{Ad}_G$ invariance of the localizing action, we can reduce bosonic critical manifold $M_\gamma$ to a representative $\gamma$ and according to \eqref{eq:tan} we decompose 
$$X(\tau)=X^{\perp}(\tau)+ u^\al (e^{- i\omega_{\ex{\al,\nu}}\tau}-1)E_\al,$$
where  the vector $X^{\perp}(\tau)$ is orthogonal to $M_{\ga}$ at $\ga$, and for its modes in addition to constraints above, we have $z_{\al,0}=z_{\al,-\ex{\al,\nu}}$ for $\ex{\al,\nu}\neq 0$. Similarly, we decompose
$$\psi(\tau)=\eta(\tau)+\chi(\tau),$$ 
 where $\chi$ is given by \eqref{psi-sol} and $\eta$ is orthogonal to $\chi$, so it does have fermion modes $\psi_{\al,0},\psi_{\al,-\ex{\al,\nu}}$ and $\psi_{j,0}$.

In the limit $s\rightarrow \infty$ the path integral $I(s)$ effectively reduces onto a small tubular neighborhood of $\mathcal{M}_{\ga}$ in $\Omega G$. Indeed, rescale $X^{\perp}$ and $\eta$ such that $X^{\perp}\rightarrow {X^{\perp}\ov \sqrt{s}}$ and $\eta\rightarrow {\eta\ov \sqrt{s}}$, which according to Remark \ref{rescaling} does not change the integration measure. Next, consider the terms in the action left over in the $s\rightarrow \infty$ limit\footnote{We could use this procedure for the derivation of the Eskin trace formula in Section \ref{E on G}. Since the critical points in that case are isolated, we obtain the same path integral \eqref{eq:inter}.} and use the delta-function constraints $\del( \sum_{n\in \mathbb Z }z_{j,n})$, $\del (\sum_{n\in \mathbb Z} z_{\al, n})$.  The
integral over modes describing the orbit $\mathcal{O}_{\ga}$ naturally  gives the  volume $\mathrm{vol}(\mathcal O_\gamma)$ with respect to the Cartan-Killing volume form on $G$, and we obtain
\beq\label{main-integral}
\lim_{s\rightarrow \infty}I(s)=\sum_{\gamma\in \Gamma/W} \text{vol}(\mathcal O_\gamma ) e^{-S_{\gamma}}\bm\int e^{-S^{loc}_{\gamma}} \mathscr DY \mathscr D\eta  \, d \chi,
\eeq 
where $S^{loc}_{\ga}$ is the classical action, 
\begin{gather*} %\label{action-Y-0}
S^{loc}_{\gamma} =\nonumber
\\ \int\left( {1\ov 2}\ex{D\dot Y,D\dot Y}+{1\ov 2}\ex{\chi,\dot \chi}+{1\ov 2}\ex{\dot \eta,D\dot \eta}+\frac{1}{2}\ex{\dot \eta,[DY,\dot \chi]}+{1\ov 4}\ex{\dot \chi, [[DY,Y],\dot \chi]}\right)d\tau.
\end{gather*}
The tangent vector $Y$ does not contain modes $z_{\al,0},z_{\al,-\ex{\al,\nu}}$ and $z_{j,0}$, and the integration over $\chi$ is finite-dimensional. 

It is tempting to make a shift $\dot\eta\mapsto\dot\eta+[Y,\dot\chi]$ in the path integral \eqref{main-integral}. We write $[Y,\dot \chi]$ explicitly as
$$ %\beq\label{Y-chi}
\left[\sum_{j,n\neq 0 }  z_{j,n}H_j e^{i\omega_n \tau} +\sum_{\al, n\neq 0,-\ex{\al,\nu}} z_{\al,n}E_\al e^{i\omega_n\tau}, -i\sum_{\be}\omega_{\ex{\be,\nu}}\psi_{\be,-\ex{\be,\nu}}E_\be e^{-i\omega_{\ex{\be,\nu}}\tau} \right].
$$ %eeq

This commutator does not contain modes $H_j$ and $E_\al e^{-i\omega_{\ex{\al,\nu}}}$. First, there is no constant term $H_j$, since it could only comes from $[E_\al,E_{-\al}]e^{i\omega_{n+\ex{\al,\nu}}\tau}$, but $n+\ex{\al,\nu}\neq 0 $. Similarly, the term $E_\al e^{-i\omega_{\ex{\al,\nu}}}$ cannot appear: it can come only from $[H_j,E_{\al}]e^{i\omega_{n-\ex{\al,\nu}}\tau}$ or from $[E_{\al},E_{\be}]e^{i\omega_{n-\ex{\be,\nu}}\tau}$, but $n\neq 0$. 

However, this commutator $[Y,\dot\chi]$ contains modes $\psi_{\al,-\ex{\al,\nu}}$ with $\ex{\al,\nu}\neq 0$. Namely, the appear in the constant term
\begin{gather*}\zeta=\frac{1}{\be}\int_{0}^{\be}[Y,\dot\chi]d\tau \\=\sum_{j}\sum_{\al\in R} \omega_{\ex{\al,\nu}}\al_j z_{j,\ex{\al,\nu}} \psi_{\al,-\ex{\al,\nu}} E_\al+\sum_{\al,\be\in R,\al+\be\neq 0} \omega_{\ex{\be,\nu}} N(\al,\be) z_{\al, \ex{\be,\nu}}  \psi_{\be,-\ex{\be,\nu}} E_{\al+\be}.
\end{gather*}
Thus the correct change of variables is $\dot\eta\mapsto\dot\phi=\dot\eta+[Y,\dot\chi]-\zeta$. As a result, we get the action (renaming $\phi$ to $\eta$ again)
\beq \label{action-Y-1}
S^{loc}_{\ga} ={1\ov 2}\int\Big( & \ex{D\dot Y,D\dot Y}+\ex{\chi,\dot \chi}+\ex{\dot \eta,D\dot \eta} 
\\&+\ex{\zeta,\mathrm{ad}_{J_{cl}}\zeta}-\ex{[Y,\dot\chi],[DY,\dot\chi]}+\frac{1}{2}\ex{\dot \chi, [[DY,Y],\dot \chi]}\Big)d\tau. 
\eeq

We have
 \begin{align*}
 [[DY,Y]\dot\chi]= & -[[Y,\dot\chi],DY]-[[\dot\chi,DY],Y]\\
 & = [DY, [Y,\dot\chi]] + [Y, [\dot\chi,DY]],
 \end{align*}
so
\begin{align*}
\ex{\dot\chi,[[DY,Y],\dot\chi]} & =\ex{\dot\chi, [DY, [Y,\dot\chi]]}+\ex{\dot\chi, [Y, [\dot\chi,DY]]}\\
& =-\ex{[DY,\dot\chi], [Y,\dot\chi]}-\ex{[Y,\dot\chi], [\dot\chi,DY]}\\
&=2\ex{[Y,\dot\chi],[DY,\dot\chi]},
\end{align*}
and we finally obtain
\begin{align}\label{S-final}
S^{loc}_{\ga} &={1\ov 2}\int\left( \ex{D\dot Y,D\dot Y}+\ex{\chi,\dot \chi}+\ex{\dot \eta,D\dot \eta} + \ex{\zeta,\mathrm{ad}_{J}\zeta}\right)d\tau.
\end{align}

\begin{remark} The main reason that the shift $\dot \eta\rightarrow \dot \eta + [Y,\dot \chi]$ could be dangerous is because $[Y,\dot \chi]$ may contain a constant Fourier mode which cannot be integrated in $L\frak{g}$.
Namely, consider the following elementary example 
\beq
I=\bm\int_{\Pi \Omega \mathfrak g} e^{ -\int_0^\be   \ex{\dot \eta, [A,\dot \eta] }d\tau }\cD\eta=\text{Pf}(\pa^{2}_{\tau}\text{ad}_A),
\eeq
where $A\in \mathfrak g$ is a constant, and consider the shift  $\dot \eta\rightarrow \dot\vp=\dot \eta+\varepsilon$ where $\varepsilon \in \Pi \mathfrak g$ is a constant. 
Obviously, after integration $\vp$ is no longer in $L\frak{g}$. Nevertheless, if we naively assume that the path integral measure does not change, $\cD\eta=\cD\vp$ ,and integration is still over $\Pi\Omega\frak{g}$, then we obtain
\beq
\int_0^\be  \ex{\dot \eta, [A,\dot \eta] }d\tau = & \int_0^\be  \ex{\dot \vp, [A,\dot \vp] }d\tau-2\int_0^\be  \ex{\dot \vp [A,\varepsilon] }d\tau+\int_0^\be  \ex{\varepsilon, [A,\varepsilon ] }d\tau
\\&=\int_0^\be  \ex{\dot \vp, [A,\dot \vp] }d\tau+\beta\ex{\varepsilon, [A,\varepsilon ] },
\eeq
where the middle term vanishes since it is a total derivative. Now since $\ex{\varepsilon, [A,\varepsilon ] }$ is a non-zero constant, we get that  under such change of variables the path integral gets multiplied by
$e^{-\beta\ex{\varepsilon, [A,\varepsilon ] }}$, which is a contradiction.
\end{remark}

To summarize, the integral over $\cD Y$ is Gaussian and everything reduces to evaluating the regularized determinant of  the following non-local linear operator 
\beq\label{E-Value}
D^{2}\ddot{Y}(\tau)+\frac{1}{\be}\left[\dot\chi(\tau),\int_{0}^{\be}[J,[Y(\xi),\dot\chi(\xi)]]d\xi\right]=\lambda Y(\tau),
\eeq
which will contain a finite product of eigenvalues that depend on the modes of $\chi(\tau)$. Note that because of the conditions on $Y(\tau)$, this operator should be considered acting on the Hilbert space $L^{2}_{\ga}([0,\be],\frak{g})$ of $\frak{g}$-valued functions, whose Fourier 
series do not contain coefficients  $z_{\al,0},z_{\al,-\ex{\al,\nu}}$ and $z_{j,0}$.

The second term in \eqref{E-Value} can be written in a more suggestive form. Namely, let $P$ be a projection on a constant term operator,
$$PY=\frac{1}{\be}\int_{0}^{\be}Y(\xi)d\xi.$$
Using equation $D\chi=0$, we have
$$[J,[Y,\dot\chi]]=[Y,[J,\dot\chi]]-[\dot\chi,[J,Y]]=-[Y,\ddot\chi]-[\dot\chi,[J,Y]]$$
and integrating over $\xi$ we get
$$\int_{0}^{\be}[J,[Y(\xi),\dot\chi(\xi)]]d\xi=-\int_{0}^{\be}[\dot\chi(\xi), DY(\xi)]d\xi,$$
so we can write our operator as
\beq\label{Operator-L}
L=D^{2}\pa^{2}_{\tau}-\mathrm{ad}_{\dot\chi}P\mathrm{ad}_{\dot\chi}D=(D\pa^{2}_{\tau}-\mathrm{ad}_{\dot\chi}P\mathrm{ad}_{\dot\chi})D.
\eeq

This means that one needs to use another basis in the $Y$-space, adapted to the commutativity of $D$ and $\mathrm{ad}_{\dot\chi}$.

Therefore, after integrating over $Y$ and $\eta$, we finally obtain the following result.
\begin{theorem}  \label{Singular-TF} The following trace formula holds. 
\begin{gather*}
e^{-\frac{1}{12}\be R} \Tr e^{-\frac{1}{2}\be \Delta}=\\
 V_G \sum_{\gamma\in\Gamma/W} {\text{vol}(\mathcal O_\gamma)\ov (2\pi)^{(d_G+d_{\mathcal O_\gamma } ) / 2}} e^{-S_\gamma} \int {\mathrm{Pf}(D\pa^{2}_\tau) \ov 
 \sqrt{\det(D^{2}\pa^{2}_{\tau}-\mathrm{ad}_{\dot\chi}P\mathrm{ad}_{\dot\chi} D)} } e^{-{1\ov 2}\int \ex{\chi,\dot \chi}d\tau }d \chi.
\end{gather*}
where the Pfaffian and determinant are taken with respect to the real Hilbert space $L^{2}_{\ga}([0,\be],\frak{g})$ and integration over $\chi $ is over the finite-dimensional space $\Pi \frak{g}/c_{\ga}$ defined in \eqref{psi-sol}.
\end{theorem}

\begin{remark}
As an example, we apply Theorem \ref{Singular-TF} in the case of $G=\mathrm{SU}(2)$. For the spectral side, we have
$$
e^{-\frac{1}{12}\be R} \Tr e^{-\frac{1}{2}\be \Delta}=\sum_{n=0}^{\infty} (n+1)^2 e^{-{\be (n+1)^2\ov 16}} 
$$

Examining the localization side, there are two distinct classes of critical orbits, distinguished by $\gamma=0$ or $\gamma=2\pi i n\sigma_3$ with $n\in\mathbb N$.

For the trivial orbit $\nu=0$, there is no $\chi$ integral and therefore
$$
{\mathrm{Pf}(D\pa^{2}_\tau) \ov 
 \sqrt{\det(D^{2}\pa^{2}_{\tau}-\mathrm{ad}_{\dot\chi}P\mathrm{ad}_{\dot\chi} D)} } ={1\ov \text{Pf}(\pa_\tau)}={1\ov \be^{3/2} }
$$

For the non-trivial critical orbit $\nu\neq 0$, the classical action for the critical orbit $\gamma=2\pi i n\sigma_3$ is $S_\gamma={16\pi^2n^2\ov \be}$ and the fermionic integral gives
$$
 \int {\mathrm{Pf}(D\pa^{2}_\tau) \ov 
 \sqrt{\det(D^{2}\pa^{2}_{\tau}-\mathrm{ad}_{\dot\chi}P\mathrm{ad}_{\dot\chi} D)} } e^{-{1\ov 2}\int \ex{\chi,\dot \chi}d\tau }d \chi={\be-32\pi^2 n^2 \ov 2\be^{5/2}}.
$$
Together with the volume factor $V_G=32\sqrt{2}\pi^2$ and $\text{vol}(\mathcal O_{\gamma\neq 0})=8\pi$,  we get the expression for the localization side
$$
{16\sqrt{\pi} \ov \be^{3/2}}+\sum_{n=1}^{\infty}{32 \sqrt{\pi}(\be-32n^2\pi^2 )\ov\be^{5/2} } e^{-{16\pi^2n^2\ov \be}} =\sum_{n=-\infty}^{\infty}{16\sqrt{\pi}(\be-32n^2\pi^2 )\ov\be^{5/2} } e^{-{16\pi^2n^2\ov \be}}
$$
This expression  matches precisely the spectral side, which can be shown using the Poisson summation formula.
\end{remark}
\begin{appendix}
\section{Semi-simple Lie groups and algebras} \label{appsec:g}
Here we present, in a succinct form, basic necessary facts on semi-simple Lie groups and algebras (see \cite{helgason1979differential, kirillov2008introduction} for details and proofs).

\subsection{Basic facts on $\frak{g}$}  \label{appsec:g.1}  Let $\frak{g}$ be a semi-simple Lie algebra of dimension $n$, let $\frak{g}^{*}$ be the dual vector space to $\frak{g}$, and let $B$ be the Killing form on $\frak{g}$,
$$B(u,v)=\Tr(\mathrm{ad}_{u}\circ\mathrm{ad}_{v}),\quad u,v\in \frak{g}.$$ 

If the Killing form $B$ is negative-definite, the semi-simple Lie algebra $\frak{g}$ is called compact. Let $\frak{g}$ be compact Lie algebra, $\frak{t}$ be its
Cartan subalgebra of dimension $r$, and let $\frak{g}_{\CC}= \frak{g}\otimes_{\RR} \CC$ be it complexification, a complex semi-simple Lie algebra with the Cartan subalgebra $\frak{t}_{\CC}= \frak{t}\otimes_{\RR} \CC$.

It will be convenient to denote by $\la~,~\ra$ the Euclidean inner product on real vector space $i\frak{t}$, defined by the Killing form, and to use the same notation for the induce inner product the $\RR$-dual vector space $i\frak{t}^{*}$.
Moreover, we will use $\la~,~\ra$ to denote the natural pairing between $i\frak{t}^{*}$ and $i\frak{t}$, $\la\al,h\ra=\al(h)$, where $\al\in i\frak{t}^{*}$ and $h\in i\frak{t}$. There will be no confusion since we will always
specify the spaces to which arguments of $\la~,~\ra$ belong.

Let $R$ be the root system for the pair $(\frak{g}_{\CC},\frak{t}_{\CC})$, so that
$$\frak{g}_{\CC}=\frak{t}_{\CC}\oplus\bigoplus_{\al\in R}\frak{g}_{\al},$$
where
$$\frak{g}_{\al}=\{x\in\frak{g}_{\mathbb C} : [h,x]=\al(h)x\;\;\text{for all}\;\; h\in\frak{t}_{\CC}\}$$
and 
$$R=\{\al\in i\frak{t}^{*} - \{0\}: \frak{g}_{\al}\neq 0\}.$$
For every root $\al\in R$ denote by $\al^{\vee}\in i\frak{t}$ the corresponding coroot, defined by
$$\la \al^{\vee},\be\ra=2\frac{\la\al,\be\ra}{\la\al,\al\ra}\quad\text{for all}\quad \be\in i\frak{t}^{*}.$$

A basis of $\frak{g}_{\CC}$ consisting of a basis of $\frak{t}_{\CC}$ and bases of $\frak{g}_{\al}$ is called Cartan-Weyl basis. Orthonormal Cartan-Weyl basis is given by  $H_{j}\in i\frak{t}$ and bases $E_{\al}$ of $\frak{g}_{\al}$, satisfying 
$$\la H_i,H_j\ra=\delta_{ij},\quad\la E_\al,E_{\be}\ra=\delta_{\al,-\be}$$
and having commutation relations
$$[H_i,E_\al]=\al(H_{i}) E_\al,\qquad [E_\al,E_{-\al}]=\sum_{i=1}^{r}\al(H_{i})H_{i},$$
and
 $$[E_\al,E_\be]= N_{\al,\be}E_{\al+\be}\quad\text{if}\quad 0\neq \al+\be\in R.$$ 
 The compact Lie algebra $\frak{g}$ is the following $\RR$-linear subspace of $\frak{g}_{\CC}$
 \begin{equation}\label{g-real}
 \frak{g}=\bigoplus_{j=1}^{r}\RR\cdot iH_{j}\oplus\bigoplus_{\al\in R_+}\RR\cdot(E_{\al}-E_{-\al})\oplus\bigoplus_{\al\in R_+}\RR\cdot i(E_{\al}+ E_{-\al}),
 \end{equation}
 since restriction of the Cartan-Killing form to it is negative definite.

A choice of element $h\in\frak{t}$ with $\al(h)\neq 0$ for all $\al\in R$ determines the set $R_{+}$ of positive roots.
The set of simple roots $S\subset R_{+}$ generates the root lattice $Q$ in $i\frak{t}^{*}$ and the corresponding set of coroots generates the coroot lattice $Q^{\vee}$ in $i\frak{t}$. The weight lattice $P$ in $i\frak{t}^{*}$ is the dual lattice to $Q^{\vee}$, and a
weight $\lambda\in P$ is called dominant if $\la\lambda,\al\ra>0$ for all $\al\in R_{+}$.  The set of dominant weights $P_{+}$ is isomorphic to $\mathrm{Irrep}\,G$, and  
for $\lambda\in P_{+}$ we have
$$C_{2}(\lambda)=\la\lambda +\rho,\lambda+\rho\ra -\la\rho,\rho\ra,\quad\text{where}\quad \rho=\frac{1}{2}\sum_{\al\in R_{+}}\al\in P_{+}.$$

Finally, an element $h\in \mathfrak t$ is called regular, if $\al (h)\notin  2\pi i \mathbb Z $ for all $\al\in R$. The following statement is crucial for the derivation of the Eskin trace formula on $G$. 

\begin{lemma} \label{lemma:exp}
If $X\in \mathfrak t$ is a regular element and  $Y\in \mathfrak g$ satisfies $e^X=e^Y$, then $Y\in \mathfrak t$.
\end{lemma}

\begin{proof} Indeed, consider a spectral decomposition of the operator $\text{ad}_{X}$, $X\in\frak{t}$, on the vector space $\frak{g}_{\CC}$. Invariant subspaces of $\mathrm{ad}_{X}$
are $\frak{t}_{\CC}$ with the eigenvalue $0$, and $\frak{g}_{\al}$ with the eigenvalues $\al(X)=\la\al,X\ra$, $\al\in R$. Using the formula $\mathrm{Ad}_{e^{u}}=e^{\mathrm{ad}_{u}}$ for $u\in\frak{g}$, we see that the operator $\text{Ad}_{e^{X}}$ has the eigenvalues $1$ and $e^{\al(X)}$. For regular $X$ we have $\al(X)\notin 2\pi i\ZZ$ for all $\al\in R$, so the eigenspace of  $\mathrm{Ad}_{e^{X}}$ corresponding to the eigenvalue $1$ is $\frak{t}_{\CC}$. Since $e^{X}=e^{Y}$ and $\text{Ad}_{e^{Y}}(Y)=Y$, we obtain that $Y\in\frak{t}$.
\end{proof}

%
%
%
%
%\vspace{10mm}
%{\color{blue}First, $e^{X}=e^{Y}$ implies $e^{X}=e^{e^{sY}X e^{-sY} } $ for all $s\in\mathbb R$. {\bf Why this formula holds?} Now suppose that $[X,Y]\neq 0$, so $\text{Ad}_{e^{sY} } X \neq X$, at least for some sufficiently small punctured neighborhood of $s=0$. This means that the exponential map is singular at $X$. Now the differential of the exponential map at $X$ is given by 
%
%\beq
%d\exp _X =e^X{ 1-e^{\text{ad}_X } \ov \text{ad}_X },
%\eeq
%and it is singular iff $\text{ad}_X$ has an eigenvalue $2\pi i \mathbb Z_{\neq 0}$. This is a contradiction since $X$ is regular. This proves that $[X,Y]=0$. Since $X\in\mathfrak t$ is regular, this implies that $Y\in\mathfrak t$.  }

\subsection{Basic facts on $G$}\label{a-2}
Let $G$ be connected semi-simple Lie group of dimension $n$, $\frak{g}$ be its Lie algebra, and let $B$ be the Killing form on $\frak{g}$.
Denote by $L_{g}:G\to G$ and $R_{g}:G\to G$ corresponding left and right translations, $L_{g}h=gh$ and $R_{g}h=hg$, $h\in G$, and denote by
$\theta$ and $\tilde\theta$ corresponding left and right invariant Maurer-Cartan forms on $G$,
$$\theta_{g}(v)=(L_{g^{-1}})_{*}v\quad\text{and}\quad \tilde\theta_{g}(v)=(R_{g^{-1}})_{*}v,\quad v\in T_{g}G.$$
They satisfy the Maurer-Cartan equations
\beq\label{M-C}
d\theta+\frac{1}{2}[\theta,\theta]=0\quad\text{and}\quad d\tilde\theta-\frac{1}{2}[\tilde\theta,\tilde\theta]=0.
\eeq

It is convenient to use Cartan method of moving frames (also called a tetrad formalism in general relativity). Namely,
let $T_{1},\dots,T_{n}$ be a basis of $\frak{g}$, 
$$[T_{a},T_{b}]=f^{c}_{ab}T_{c},$$ 
where $f^{c}_{ab}$ are the structure constants, and let $\bm{x}=(x^{1},\dots,x^{n})$ be local coordinates in the neighborhood of $g\in G$. We have
$$\theta=\theta^a_\mu(\bm{x}) T_{a} dx^\mu\quad\text{and}\quad \tilde\theta=\tilde\theta^a_\mu(\bm{x}) T_{a} dx^\mu,$$
so corresponding bases of left-invariant and right-invariant vector fields on $G$ are given by 
\beq\label{L-R-0}
e_{a}=\theta^{\mu}_{a}(\bm{x})\pa_{\mu}\quad\text{and}\quad  f_{a}=\tilde\theta^{\mu}_{a}(\bm{x})\pa_{\mu},\quad\text{where}\quad \pa_{\mu}=\frac{\pa}{\pa x^{\mu}}.
\eeq
Here $\theta^{\mu}_{a}(\bm{x})$ and  $\tilde\theta^{\mu}_{a}(\bm{x})$ are inverse matrices to  $\theta^{a}_{\mu}(\bm{x})$ and  $\tilde\theta^{a}_{\mu}(\bm{x})$. 
For $v=v^{\mu}\pa_{\mu}\in T_{g}G$ we put $v_{L}=\theta_{g}(v)\in \frak{g}$ and $v_{R}=\tilde\theta_{g}(v)\in\frak{g}$, so
$$v_{L}=v_{L}^{a} T_{a}\quad\text{and}\quad v_{R}=v_{R}^{a}T_{a},\quad\text{where}\quad v_{L}^{a}=\theta^{a}_{\mu}(\bm{x})v^{\mu}\quad\text{and}\quad  v_{R}^{a}=\tilde\theta^{a}_{\mu}(\bm{x})v^{\mu}.$$
Equivalently, $v=v_{L}^{a}e_{a}=v_{R}^{a}f_{a}$.

For compact $G$ the bilinear form $-B$ is positive-definite and determines a bi-invariant Riemannian metric $(~,~)$ on $G$, the Cartan-Killing metric
\beq\label{Riemann}
(u,v)_{g}=-B(\theta_{g}(u),\theta_{g}(v)),\quad u,v\in T_{g}G.
\eeq
In local coordinates we have
$$(u,v)_{g}=g_{\mu\nu}(\bm{x})u^{\mu}v^{\nu},\quad\text{where}\; v=v^{\mu}\pa_{\mu},\quad u=u^{\mu}\pa_{\mu}, $$
and it follows from \eqref{Riemann} that
\beq\label{tetrad}
g_{\mu\nu}(\bm{x})=g_{ab}\theta^{a}_{\mu}(\bm{x})\theta^{b}_{\nu}(\bm{x}),\quad\text{where}\quad g_{ab}=-B (T_{a},T_{b}). 
\eeq

\subsection{Coadjoint orbits of $LG$} \label{A3} Let $LG$ be the loop group of $G$ and $\widehat{L\frak{g}}=L\frak{g}\oplus\RR$ be the central extension of the Lie algebra $L\frak{g}$ by the $2$-cocycle
$$c(X,Y)=\frac{1}{\beta}\int_{0}^{\be}\la \dot{X}(\tau),Y(\tau)\ra d\tau.$$
The adjoint action of $LG$ on $\widehat{L\frak{g}}$ is given by
$$\mathrm{Ad}_{g}(X,\al)=\left(\mathrm{Ad}_{g}(X), \al + \frac{1}{\beta}\int_{0}^{\be}\la J(\tau),X(\tau)\ra d\tau\right),\quad J=J(g)=g^{-1}\dot{g}.$$

The non-degenerate bilinear form on $\widehat{L\frak{g}}$,
$$\la(X_{1},\al_{1}), (X_{2},\al_{2})\ra=\frac{1}{\be}\int_{0}^{\be}\la X_{1}(\tau), X_{2}(\tau)\ra d\tau + \al_{1}\al_{2},$$
allows to identify $\widehat{L\frak{g}}^{*}$ with $\widehat{L\frak{g}}$. Under this identification, the coadjoint action of $LG$ on $\widehat{L\frak{g}}^{*}$ takes the form
\begin{equation}\label{co-ad}
\mathrm{Ad}^{*}_{g}(X,\al)=(\mathrm{Ad}_{g}(X) -\al\,\dot{g}g^{-1},\al),\quad (X,\al)\in \widehat{L\frak{g}}.
\end{equation}

The Lie algebra $L\frak{g}$ is naturally identified with the hyperplane $\al=1$ in $ \widehat{L\frak{g}}$, and
since $-\dot{g}g^{-1}=J(g^{-1})$,  it follows from \eqref{co-ad}  that the image of the map $J: LG\to L\frak{g}$ is a coadjoint orbit $\widehat{\mathcal{O}}_{1}$ of the element $(0,1)\in\widehat{L\frak{g}}$.

\section{Symplectic geometry of $T^{*}G$} \label{appsec:symp}
Here $G$ is a Lie group,  $\frak{g}$ is its Lie algebra, and $\frak{g}^{\ast}$ is  
the dual vector space to $\frak{g}$ with the natural pairing which is this section we denote by $(~,~): \frak{g}^{*}\times\frak{g}\to\RR$. 
\subsection{The symplectic form}
\begin{itemize}
\item The canonical $1$-form $\vartheta$ on $T^{*}G$ (the Liouville $1$-form on a cotangent bundle) is defined as follows. Let $\pi:T^{*}G\to G$ be the natural
projection and $\pi_{*}:T(T^{*}G)\to TG$ be its differential. Then for $\xi\in T_{(g,p)}T^{*}G$, where $g\in G$ and $p\in T^{*}_{g}G$, we have
\beq\label{L-one}
\vartheta(\xi)=p(v),\quad\text{where}\quad v=\pi_{*}(\xi)\in T_{g}G.
\eeq
\item The canonical symplectic form on $T^{*}G$ is defined by $\omega=d\vartheta$.

\item The left translations (and also the right translations)
trivialize the tangent and cotangent bundles to $G$, 
$$TG\simeq G\times\frak{g}\quad \text{and}\quad T^{\ast}G\simeq G\times\frak{g}^{\ast},$$ 
so putting  in \eqref{L-one} $p=(L^{*}_{g})^{-1}\alpha$, where $\alpha\in\frak{g}^{*}$, we obtain
\beq\label{L-two}
\vartheta(\xi)=(\alpha, (L_{g^{-1}})_{*}v).
\eeq
Using Maurer-Cartan form $\theta$, 
\eqref{L-two} can be written as
$$\vartheta_{(g, \alpha)}=(\alpha, \theta_{g}),$$
so using the Maurer-Cartan equation, we obtain a simple formula for the symplectic form $\omega$
\beq\label{omega-G}
\omega=(d\al,\theta)+(\al,d\theta)=(d\al,\theta)-\frac{1}{2}(\al,[\theta,\theta]).
\eeq
Explicitly, for $X_{1}=(v_{1},\al_{1}), X_{2}=(v_{2},\al_{2})\in T_{(g,\al)}T^{*}G\simeq T_{g}G\times\frak{g}^{*}$ 
 we have\footnote{Note that due to the definition of the exterior product of two $1$-forms, the factor $\frac{1}{2}$ in the Maurer-Cartan formula becomes $1$.}
\beq\label{omega-G-1}
\omega(X_{1},X_{2})=(\al_{1},\theta(v_{2}))-(\al_{2},\theta(v_{1}))-(\al,[\theta(v_{1}),\theta(v_{2})]).
\eeq
\end{itemize}
\subsection{Poisson brackets}
The Poisson bracket of  smooth functions on $T^{*}G$ is defined by
$$\{f_{1}, f_{2}\}=\omega(X_{f_{1}},X_{f_{2}}).$$
Here $X_{f}$ is the Hamiltonian vector field for a function $f$, defined by
$$i_{X_{f}}\omega=-df.$$
\begin{itemize}
\item For every $u\in\frak{g}$ consider the function $l_{u}$ on $T^{*}G$, defined by
$$l_{u}(g,\al)=(\al, u).$$
The corresponding Hamiltonian vector field $X_{l_{u}}$ at a point $(g,\al)\in T^{*}G$ is
$$X_{l_{u}}=((L_{g})_{*}u,\mathrm{ad}^{*}_{u}\al)\in T_{(g,\al)}T^{*}G,$$
where $\mathrm{ad}^{*}$ is the coadjoint action of $\frak{g}$,
$$(\mathrm{ad}^{*}_{u}\al,v)=(\al,\mathrm{ad}_{u}v)=(\al,[u,v]).$$
Indeed, $dl_{u}=u$ and for a vector field $Y=(v,\beta)$ on $T^{*}G$ we have
\begin{align*}
\omega(X_{l_{u}},Y) & =(\mathrm{ad}^{*}_{u}\al),\theta(v))-(\beta,\theta((L_{g})_{*}u))-(\al,[\theta((L_{g})_{*}u),\theta(v)])\\
& =(\mathrm{ad}^{*}_{u}\al),\theta(v))-(\beta, u)-(\al,[u,\theta(v)])\\
&=-(\beta,u)=-dl_{u}(Y).
\end{align*}
Then
\begin{align*}
\{l_{u_{1}},l_{u_{2}}\} &=\omega(X_{l_{u_{1}}},X_{l_{u_{2}}})=-(\mathrm{ad}^{*}_{u_{2}}\al,u_{1})=(\al,[u_{1},u_{2}]),
\end{align*}
so
\beq\label{PB-1}
\{l_{u_{1}},l_{u_{2}}\}=l_{[u_{1},u_{2}]},\quad u_{1}, u_{2}\in\frak{g}.
\eeq
\item Trivialization of $TG$ and $T^{*}G$ by the right translations leads to the functions $r_{u}$ on $T^{*}G$, defined by 
$$r_{u}(g,\al)=(\al, \mathrm{Ad}_{g^{-1}}u).$$
The corresponding Hamiltonian vector field is $X_{r_{u}}=((R_{g})_{*}u, 0)$, so
\beq\label{PB-1-right}
\{r_{u_{1}},r_{u_{2}}\}=-r_{[u_{1},u_{2}]}\quad\text{and}\quad\{l_{u_{1}}, r_{u_{2}}\}=0. 
\eeq
\item Another class of functions on $T^{*}G$ is given by smooth $f:G\to\RR$. It is easy to see that 
$$X_{f}=(0,-(L_{g})^{*}df).$$
Indeed, for $Y=(v,\beta)$ be have
\begin{align*}
\omega(X_{f},Y) & =-((L_{g})^{*}df,\theta(v))=-(df, (L_{g})_{*}\theta(v))=-(df,v),
\end{align*}
so such functions Poisson commute and 
\beq\label{PB-2}
\{l_{u},f\}=((L_{g})^{*}df,\theta((L_{g})_{*}u))=(df,(L_{g})_{*}u) =e_{u}(f).
\eeq
\end{itemize}

\section{Free particle on a Lie group}\label{free on G}
In local coordinates on $G$ we have
\beq\label{G-free}
\LL=\frac{1}{2}g_{\mu\nu}(\bm{x})\dot{x}^{\mu}\dot{x}^{\nu},
\eeq
where $g_{\mu\nu}(\bm{x})$ is given by \eqref{tetrad}. Consider left-invaraint conserved current $J=g^{-1}\dot{g}$ 
\beq\label{G-current}
J=J^{a}T_{a},\quad J^{a}=\theta^{a}_{\mu}\dot{x}^{\mu}\quad\text{and}\quad \dot{x}^{\mu}=\theta^{\mu}_{a}J^{a},
\eeq
so the Lagrangian \eqref{G-free} can be rewritten in terms of currents
$$\LL=\frac{1}{2}g_{ab}J^{a}J^{b}.$$
By the Legendre transform
$$p_{\mu}=\frac{\pa \LL}{\pa\dot{x}^{\mu}}=g_{\mu\nu}\dot{x}^{\nu}$$ 
so using $\dot{x}^{\mu}=g^{\mu\nu}p_{\nu}$ we have
\beq\label{H-Group}
H=p_{\mu}\dot{x}^{\mu}-\LL=\frac{1}{2}g^{\mu\nu}p_{\mu}p_{\nu}.
\eeq

Using \eqref{G-current}, we can rewrite Legendre transform in terms of $J$
$$p_{\mu}=g_{ab}\theta^{a}_{\mu}J^{b}=\theta^{a}_{\mu}J_{a},\quad J_{a}=g_{ab}J^{b}.$$
Note that $J^{a}$ are functions on $TG$, so introducing the functions $l_{a}$ on $T^{*}G$ by
\beq\label{e-functions}
l_{a}=\theta^{\mu}_{a}p_{\mu},
\eeq
we have
$$p_{\mu}\dot{x}^{\mu}=g^{ab}l_{a}l_{b}\quad\text{and}\quad H=\frac{1}{2}g^{ab}l_{a}l_{b}.$$
We see that $l_{a}$ are precisely the functions $l_{u}$ on $T^{*}G$, defined in the Appendix \ref{appsec:symp}, for $u=T_{a}$.
According to \eqref{PB-1}, they have the following Poisson brackets
\beq\label{PB-l}
\{l_{a}, l_{b}\}=f^{c}_{ab}l_{c}.
\eeq
One can do exactly same exercise with a right-invariant conserved current $\tilde J=\dot g g^{-1}$ with a similarly defined function $r_a=\tilde \theta ^\mu _a p_\mu$ on $T^* G$. Namely, it follows from
\eqref{PB-1-right} that
\beq \label{PB-II}
\{r_{a}, r_{b}\}=-f^{c}_{ab}r_{c}\quad\text{and}\quad\{r_{a},l_{b}\}=0.
\eeq
\end{appendix}

\bibliographystyle{amsplain}
\bibliography{Ref.bib}

\end{document}